\documentclass{sig-alternate}
\usepackage{times}
\usepackage{verbatim}
\usepackage{graphicx}
\usepackage{multirow}
\usepackage{algorithmic}
\usepackage{times}
\usepackage{helvet}
\usepackage{courier}
\usepackage{verbatim}
\usepackage{url}
\usepackage{bm}
\usepackage{color}
\usepackage{amsmath, amssymb,sansmath}
\usepackage{balance}
\usepackage[lined,algonl,boxed]{algorithm2e}
\usepackage{longtable}
\usepackage{epstopdf}

\makeatletter
\newif\if@restonecol
\makeatother

\usepackage[lined,algonl,boxed]{algorithm2e}
\usepackage{epsfig}
\usepackage{subfigure}
\usepackage{multirow}
\usepackage{amsmath}

\newcommand{\reminder}[1]{\textbf{\color{red}[* * #1 **]}}  
\newcommand{\hide}[1]{} 
\newcommand{\vpara}[1]{\vspace{0.1in}\noindent\textbf{#1 }}
\newcommand{\para}[1]{\vspace{0.01in}\noindent\textbf{#1 }}
\newcommand{\beq}[1]{\vspace{0in}\begin{equation}#1\end{equation}\vspace{0in}}
\newcommand{\beqn}[1]{\vspace{0in}\begin{eqnarray}#1\end{eqnarray}\vspace{0in}}
\newcommand{\besp}[1]{\begin{split}#1\end{split}}

\newcommand{\model}{\textbf{M$^3$D}}

\newdef{lemma}{Lemma}
\newdef{definition}{Definition}
\newdef{problem}{Problem}
\newdef{theorem}{Theorem}

\begin{document}

\title{
Modeling the Interplay Between Individual Behavior and Network Distributions
}

%

\author{
\alignauthor Yang Yang$^{\dag}$, Jie Tang$^{\dag}$, Yuxiao Dong$^{\ddag}$, Qiaozhu Mei$^{\star}$, Reid A. Johnson$^{\ddag}$, Nitesh V. Chawla$^{\ddag}$\\
 \affaddr{$^{\dag}$Department of Computer Science and Technology, Tsinghua University}\\
 \affaddr{$^{\ddag}$Department of Computer Science and Engineering, University of Notre Dame}\\
 \affaddr{$^{\star}$Department of Electrical Engineering and Computer Science, University of Michigan}\\
\email{\sf yyang.thu@gmail.com, jietang@tsinghua.edu.cn, \{ydong1, rjohns15, nchawla\}@nd.edu, qmei@umich.edu}
}


\maketitle
\sloppy
\begin{abstract}

It is well-known that many networks follow a  power-law degree distribution; however, the factors that influence the formation of their distributions are still unclear. 
How can one model the connection between individual actions and network distributions?
How can one explain the formation of group phenomena and their evolutionary patterns?

In this paper, we propose a unified framework, \textbf{M$^3$D}, to model human dynamics in social networks from three perspectives: macro, meso, and micro. 
At the micro-level, we seek to capture the way in which an individual user decides whether to perform an action.
At the meso-level, we study how group behavior develops and evolves over time, based on individual actions.
At the macro-level, we try to understand how network distributions such as power-law (or heavy-tailed phenomena) can be explained by group behavior. 
We provide theoretical analysis for the proposed framework, and discuss the connection of our framework with existing work.

The framework offers a new, flexible way to explain the interplay between individual user actions and  network distributions, and can benefit many applications. 
To model heavy-tailed distributions from partially observed individual actions and to predict the formation of group behaviors, we apply \textbf{M$^3$D} to three different genres of networks: Tencent Weibo, Citation, and Flickr. 
We also use information-burst prediction as a particular application to quantitatively evaluate the predictive power of the proposed framework. 
Our results on the Weibo indicate that \model's prediction performance exceeds that of several alternative methods by up to 30\%.

\hide{
How to model the dynamic social networks is a critical research task. Generally, the goal of this study is to describe how an object evolves in a given social network. For instance, the diffusion process of a message propagated in the network, the growth of the relationships between users, and so on. More specifically, we are interested in answering the following three questions: (1) at micro-level, how each individual behaves at each time unit; (2) at meso-level, how groups of individuals' behaviors effect the object at each time unit; and (3) at macro-level, what properties the object finally have. Although numerous attempts have been made for this study, for instance, power-law models are used to describe the object's macro-level properties (e.g., degree distribution for relationship analysis) frequently, model the macro-level frequently, few of them can answer the above three questions in a unified framework. 

In this paper, we propose the M3D framework to describe the evolution process of a given object in a social network. More specifically, we use information diffusion as an example to describe how each individual behaves and effects the diffused information, how groups of users behave, and the diffusion phenomena in the whole social network. We evaluate the proposed framework in a real social networking data.
}

\end{abstract}




\section{Introduction}
\label{sec:intro}

Statistics show that 1\% of Twitter users produce 50\% of its content~\cite{Wu:WWW2011} and control 25\% of its information diffusion~\cite{lou2013mining}, while 5\% of Wikipedia contributors generate 80\% of its content~\cite{Muchnik:SR14}. 
Such phenomena have received much attention, and several state-of-the-art models~\cite{barabasi1999emergence,clauset2009power,faloutsos1999power,newman2005power} have been proposed, to explain their underlying mechanism.
However most of these studies focus on modeling the totality of interactions between individuals, while ignoring the temporal aspects of individual actions~\cite{Song:Connect2012}. 
Yet the dynamics of social phenomena are, at a fundamental level, driven by individual user actions~\cite{vazquez2006modeling}, resulting in a clear and present need for understanding the connection between human dynamics and network distributions. 

The connection between human dynamics (e.g., e-mail communication between individuals) and network distributions has been studied in physics~\cite{Song:Connect2012,vazquez2006modeling}, economics~\cite{black1973pricing}, and sociology~\cite{Bernheim:94,Kelman:58}. 
V{\'a}zquez et al.~\cite{vazquez2006modeling} showed that the timing of individual user actions follows a non-Poisson distribution pattern, and the ``bursty'' nature of human behavior can be modeled based on the decisions of individual user.
Rybski et al.~\cite{Rybski:SR12} studied group behavior in social communities, and tried to understand the origin of clustering and long-term persistence.
Muchnik et al.~\cite{Muchnik:SR14} tried to understand how network distributions such as degree distribution (power law) arises from individual actions. They found that action and degree are not strongly correlated.
However, these studies do not provide explicit explanations for the connection between individual actions and network distributions.
Recently, Song et al.~\cite{Song:Connect2012} focused on studying communication patterns between users using mobile, e-mail, Twitter, instant message data. They discovered a series of interesting relationships that quantitatively connect human dynamics to several properties of the network.

\hide{
Such dramatic dominations by handful of activities have been referred to as the heavy-tailed phenomena and discovered to be ubiquitous in a variety of network systems~\cite{newman2005power,clauset2009power}. 
For example, the degree of nodes in the Internet Infrastructure, the World Wide Web (WWW), and the Facebook social network demonstrate a power law distribution~\cite{faloutsos1999power,barabasi1999emergence,Ugander:FB2014}. 
Or the number of indexed pages at a given site is lognormally distributed for every timestamp~\cite{huberman1999internet}. 
Or, beyond power law and lognormal distributions, the double Pareto lognormal distributions---also skewed and heavy-tailed---are revealed from mobile call graphs~\cite{Christos:KDD08}.
Accordingly, each network as a whole presents remarkable statistical regularities. 
}

\begin{figure*}[t]
\centering
\label{fig:introexp}
\epsfig{file=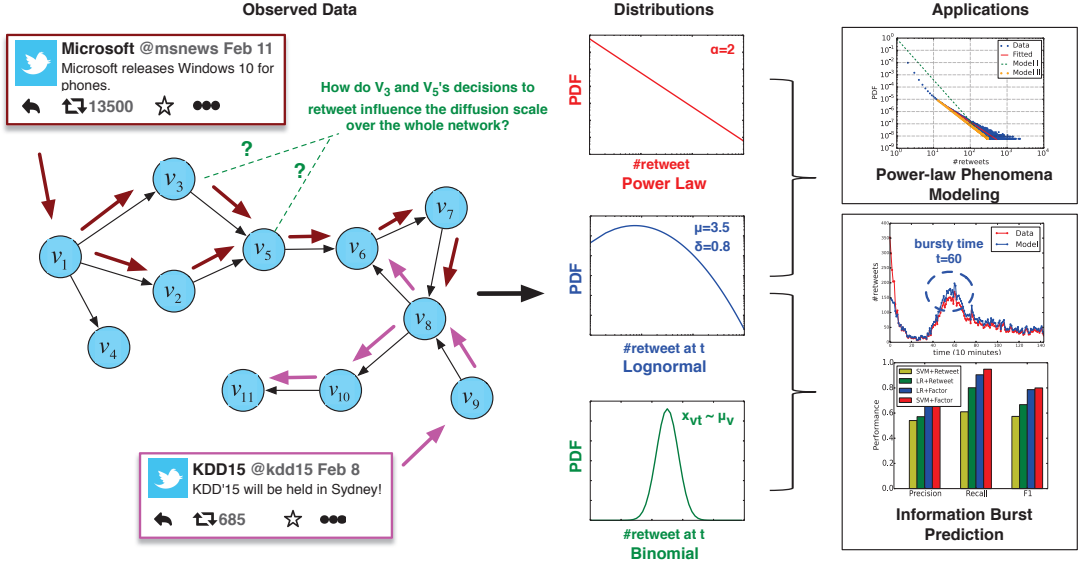, width=5.8in}
\caption{Example of modeling individual retweeting behavior and the distribution of the number of retweets. The left figure shows an example in which two tweets are diffused in Twitter by retweeting. The middle figure shows the binomial distribution that determines individual user actions, the lognormal distribution that capture the retweet count for each tweet, and the power law distribution that models the distribution of the number of retweets for all tweets. The right figure shows examples of applications.\normalsize}
\end{figure*}

In this work, our goal is to develop a theoretical framework to model human
dynamics in social networks from three perspectives: macro, meso, and micro. At the micro-level, we try to capture how individual users make a decision to perform an action (e.g., to retweet a message on Twitter). At the meso-level, we study how individual  actions develop into 
group behavior (e.g., the diffusion of a message) and how group behavior evolves over time.
At the macro-level, we investigate how the network distributions such as power-law (or heavy-tailed phenomena) arise from group behavior.

Figure~\ref{fig:introexp} illustrates the problem addressed in this paper.
The left figure shows the example in which two tweets are diffused in Twitter by retweeting. 
In the diffusion process, each individual makes a decision to retweet or not, according to a personalized binomial distribution. 
The number of retweets for each message has been modeled using a lognormal distribution, and the retweet counts for all messages follow the power-law distribution.
The right figure shows several potential applications, namely: modeling network phenomena using power-law and information-burst prediction. 
As a non-trivial problem, the fundamental challenge lies in the uncertainty of how different sub-models are intrinsically connected and how they are developed. It is also important to validate the effectiveness of such a modeling framework in real, large networks. 

\hide{
Generally, we can observe that the number of retweets of tweets, such as $xx$ and $yy$, follow a power law distribution over an observed timespan. 
Specifically, an individual who is exposed to a tweet $xx$, such as $v_3$ or $v_5$, makes a decision to retweet it or not, which can be modeled as a Binomial test. 
The group behaviors---the collection of such ``random'' retweeting actions on this tweet $xx$---actually demonstrate a lognormal distribution. 
In observing so, the integration of such individual retweeting actions--over both all users and all tweets---eventually exhibits a ``regular'' distribution in Twitter network as a whole---the heavy-tailed phenomenon.  
}


In this paper, we conclusively demonstrate the underlying mechanisms by which 
heavy-tailed phenomena develop from 
individual user actions. 
We propose a unified framework, referred to as \model, to 1) model the statistical distributions of individual actions, group behaviors, and heavy-tailed phenomena, 
and to 2) unveil the emerging process of heavy-tailed phenomena from individual actions and group behavior. 
This proposed framework produces several interesting results, both theoretical and empirical.
Theoretically, we obtain a theorem that suggests a lower bound on the individual action-adoption probability for the existence of the power-law distribution in the number of action adopters.
Empirically, by leveraging our framework, we demonstrate that it is possible to achieve an accuracy of 90\% for predicting future information bursts.

The proposed framework is flexible and can benefit many applications. To model heavy-tailed distributions from partially observed individual actions, we apply \model~to three genres of networks: Tencent Weibo\footnote{http://t.qq.com, one of the largest microblogging services in China.}, Citation, and Flickr.
We verify that our proposed framework can explain the emerging process of heavy-tailed phenomena from individual actions in these networks. 
For example, our results on the Weibo network---with more than 320 million users and 4.6 billion tweets over four months---convincingly demonstrate that 1) the retweeting action of each individual aggregates to the lognormal distribution as suggested in our framework, and 2) the lognormal distribution at each timestamp is integral to the power law distribution.
Our conclusions on the connection between individual actions and heavy-tailed phenomena in real-world networks give rise to important implications for understanding the underlying mechanisms of social emergence.





\vpara{Organization.} The rest of the paper is organized as follows. In Section~\ref{sec:framework}, we describe the unified framework we propose to model user actions at three levels of granularity, and describe how the three levels connect with each other. In Section~\ref{sec:observation}, we introduce our experimental setup based on data from three real social network. In Section~\ref{sec:exp}, we present the experimental results to validate \model. In Section~\ref{sec:related}, we review relevant related work. Section~\ref{sec:conclusion} concludes the paper.


\section{Model Framework}
\label{sec:framework}

In this section, we propose a unified framework \model \ to model the interplay between individual actions at the micro-level, group behaviors at the meso-level, and network distribution at the macro-level. 

\subsection{Formulation}
\label{sec:definition}
Let $G=(V, E)$ denote an \textit{observed} network that is a subnetwork of the \textit{complete} network $H=(V^H, E^H)$, as $H$ itself is too large to be observed entirely in practice. $V$ is the set of users and $E\subset V\times V$ is the set of edges between users. 
Our goal is to study how individual users' actions in $G$ emerge as the macroscopic phenomena in the complete network $H$. 
To begin with, we first give definitions of several concepts over a social network:
individual actions, group behaviors, and network distribution. 

\begin{definition}
\label{def:individual_action}
\textbf{Individual action.} 
An action $z$ performed by user $v$ at time $t$ is represented as $x_{tvz} \in \{-1, 1\}$. When we observe user $v$'s action $z$ at time $t$, we denote $x_{tvz}=1$; otherwise $x_{tvz}=-1$.
\end{definition}

The action can be defined differently in different social networks. For example, in Twitter, we can define the action as retweeting a tweet; in a scientific network, we define the action as citing a specific paper; and in Flickr, we define the action as posting comments to a specific photo.
%
By accumulating individual actions, we can observe group behavior at the meso-level. Formally, we have the following definition.

\begin{definition}
\label{def:collective}
\textbf{Group behavior.} 
For a given action $z$, 
we denote the number of users (adopters) who adopted the action at a specific time $t$ in the complete network $H$ as $n_z(t)$.
\end{definition}


It is worth noting that the group behavior is defined over the complete network $H$
instead of the observed network $G$. Also we use $N_z(t)$ to denote the cumulative number of adopters up to time $t$, and $N_z$ to denote  the total number of adopters until time $T\ge t$ ($T$ is larger than any observed time $t$).
Finally, at the macro-level, we consider all actions $\{z\}$ in the social network. 

\begin{definition}
\label{def:network}
\textbf{Network distribution.} 
Given a set of actions $\{z\}$, and the corresponding set of numbers $\{N_z\}$ for all the actions, $P(N_z)$ represents the network distribution of the actions.
\end{definition}

Regarding the network distribution, heavy-tailed distributions have been demonstrated to be ubiquitous in social networks~\cite{newman2005power,clauset2009power,faloutsos1999power,barabasi1999emergence}. 
In this work, we focus on modeling the \textit{heavy-tailed phenomena} as a macro-level reflection of individual actions in social networks. 
\hide{
At the macro-level, we finally study the phenomena that the network $H$ as a whole presents. 
Heavy-tailed distributions have been demonstrated to be ubiquitous in social networks~\cite{newman2005power,clauset2009power,faloutsos1999power,barabasi1999emergence}. 
Herein, we aim to model the heavy-tailed phenomena as a macro-level reflection of individual actions in social networks. 
}%
Finally, as a conclusion, an individual action is assigned to each user in the observed network $G$ and represents the state of the user (as having adopted the action or not). A group behavior is assigned to a group of users and represented as the \textit{number} of action adopters within the group. A heavy-tailed phenomenon is represented as a \textit{distribution} to describe the popularity of each action over the complete network $H$.



\vpara{Goal.} 
Given an observed network $G$ that is a subnetwork of the complete network $H$, how to unveil the mechanisms by which the heavy-tailed phenomena in network $H$ emerge from individual actions $x_{tvz}$ in network $G$? 
We propose a unified framework to model 1) individual actions, group behaviors, and heavy-tailed phenomena together;
and 2) the emerging process by which individual actions $x_{tvz}$ of users are revealed to be integral to the heavy-tailed phenomena in network $H$ as a whole.




\subsection{Modeling Individual Actions and Group Behavior}

We first introduce independent models to 
model individual actions and group behavior. 
We then demonstrate how group behavior arises from individual actions. 

For each user $v$ in the observed network $G$, and a given action $z$, it is natural to assume individual actions $x_{tvz}$ at different timestamps $t$ are independent and identically distributed (i.i.d.).
Given this, we formally define the individual model of user $v$ as follows:

\begin{definition}
\label{def:individual_model}
\textbf{Individual Model}. In the observed network $G$, each user $v$ is assigned a binomial parameter $\mu_{vz}$ for each action $z$, such that for any time $t$, we have $x_{tvz} \sim Bernoulli(\mu_{vz})$ with $P(x_{tvz}=1)=\mu_{vz}$, $P(x_{tvz}=-1)=1-\mu_{vz}$.  
\end{definition} 

We then define the following group model to connect individual actions and group behavior.

\begin{definition}
\textbf{Group Model.} At time $t$, let $y^{+}_{t}$ be the number of individual actions with $x_{tvz}=1$ and $y^{-}_{t}$ be the number of individual actions with $x_{tvz}=-1$.  The group model is then a dynamic system governing the evolution of group behavior over time:
\beq{
\label{eq:meso}
n_z(t)=n_z(t-1)U_z^{y^{+}_{t}}D_z^{-y^{-}_{t}},
}

\noindent where $U_z$ is an ``upward factor'', which describes how an individual user in $G$ adopting $z$ will influence others in $H$; $D_z$ is a ``downward factor'', which describes how an individual user who did not adopt $z$ influences the others not to adopt the action. 
\end{definition}

Additionally, $n_z(t_0)=1$, where $t_0$ is the timestamp when the first user adopts $z$. Without further explanation, we refer to $t_0$ as timestamp $0$, for the sake of simplicity. 
Please notice that both the individual model and the group model are defined based on the action $z$. We omit the subscript $z$ in $\mu_{vz}$, $U_z$ and $D_z$ in the following descriptions to keep our notations simple. 

\hide{
\begin{figure*}[t]
\centering
\label{fig:framework}
\epsfig{file=figs/framework.pdf, width=6in}
\caption{An example of our proposed framework by letting $U=D=e$. The framework exhibits lognormal distribution in group behavior modeling and power-law distribution in network phenomena modeling. }
\end{figure*}
}

\vpara{Behavior of \model: a lognormal arises.} 
Let us now examine 
our framework to see how group behavior distributes under a certain configuration of a group model's parameters. 

\begin{theorem}
With the definition of $U=D=e$, when $t$ is large enough, $n_z(t)$ converges to a lognormal random variable with a mean $t\sum_v \mu_v$ and a variance $t\sum_{v}\mu_v(1-\mu_v)$. 
\end{theorem}

\begin{proof}

We first consider another representation of $y_t$ as 

\beq{
\label{eq:another_y}
y_{t} = y^+_t - y^-_t.
}

The total number of users $m$ can be denoted as $m=y_t^++y_t^-$. Together with  Eq.~\ref{eq:another_y}, we have

\beq{
\label{eq:hd}
y^+_t=\frac{1}{2}(m+y_t), ~~~~~
y^-_t=\frac{1}{2}(m-y_t).
}

With the definition $U=D=e$, according to Eq.~\ref{eq:meso}, we have

\beq{
\besp{
n_z(t) &= n_z(t-1)\exp(\frac{m+y_t}{2})\exp(\frac{-m+y_t}{2})\\
&= n_z(0) \exp(\sum_t y_t).
}
}

This is equivalent to proof that when $t$ is large enough, $\sum_t y_t \rightarrow N(-t \times \sum_v \mu_v, -t \times \sum_v (\mu_v)(1-\mu_v))$. 
We use $y'_{v}=\sum_{t'=1}^{t} x_{vt'}$ to denote the total effects made by user $v$. As $\{x_{vt'}\}$ are i.i.d, when $t$ is large enough, according to the central limit theorem, we have 

\beq{
\label{eq:yv}
\sqrt{t}[\frac{1}{t}\sum_{t'=1}^{t}x_{vt'}-\mu_v] \sim N(0, \mu_v(1-\mu_v)).
}

Thus we obtain

\beq{
\label{eq:yv_2}
y'_v=\sum_{t'=1}^{t} x_{vt'} \xrightarrow{d} N(t\mu_v, t\mu_v(1-\mu_v))
}

and

\beq{
\label{eq:yt}
\sum_t y_t=\sum_{v=1}y'_v \xrightarrow{d} N(t \sum_v \mu_v, t \sum_v \mu_v(1-\mu_v)).
}

Therefore, $n_z(t) \propto \exp(y_t)$ converges to the lognormal distribution with mean $t \sum_v \mu_v$ and variance $t \sum_v \mu_v(1-\mu_v)$, i.e., 

\beq{
\label{eq:pn}
P(n_z(t))=\frac{1}{n_z(t)\delta \sqrt{2\pi t}}\exp(\frac{-(ln{n_z(t)-\tau t})^2}{2\delta^2 t}),
} 

\noindent where 

\beq{
\label{eq:micro2meso}
\tau =\sum_v \mu_v,~~~~~
\delta^2 =\sum_v\mu_v(1-\mu_v).
}
\end{proof}

We refer to $\tau$ and $\delta$ as the parameters of the group model. Together with time $t$, they express the lognormal distribution that arises for $n_z$. 
More generally, approximate lognormal distributions can be obtained when $U=D=e$ does not holding. Specifically, consider 

\beq{
\label{eq:micro2meso_2}
\ln n_z(t) = \ln n_z(0) + \sum_{t'=1}^{t} \ln U^{y_{t'}^+} + \sum_{t'=1}^{t} \ln D^{y_{t'}^-}.
}

According to the Central Limit Theorem, $\sum_{t'=1}^{t} \ln U^{y_{t'}^+}$ and $\sum_{t'=1}^{t} \ln D^{y_{t'}^-}$ converge to a normal distribution. As the product of lognormal distributions is again lognormal, for sufficiently large $t$, $n_z(t)$ will asymptotically approach a lognormal distribution. 


\vpara{Significance.} 
We conclusively demonstrate that the group behavior of the network---the collection of random (binomial) individual actions---follows a lognormal distribution. 
Specifically, the parameters of the lognormal distribution can be represented by our individual models, i.e., the lognormal distribution at time $t$ is parameterized with the mean as $t\sum_v \mu_v$ and the variance as $t\sum_{v}\mu_v(1-\mu_v)$, where $\mu_v$ is the binomial parameter from the individual model. 

\subsection{Modeling Heavy-Tailed Phenomena}

We study how the integration of user behavior over time eventually exhibits the heavy-tailed phenomena in the complete network. 

For each action $z$, we define $N_z$ as the total number of adopters in the complete social network $H$. Formally, let $T$ denote the observation time window; we have $N_z=\int_{t=0}^{T}n_z(t)$. We then study the distribution of $N_z$. 

As we concluded above, $n_z(t)$ converges to a lognormal variable when $t$ is sufficiently large, i.e., $\ln P(n_z(t)) \sim N(\phi,  \varphi)$, where $\phi$ and $\varphi$ denote the mean and the variance, respectively. We assume $U=D=e$ holds. Therefore, $\phi=\tau t$ and $\varphi=\delta^2 t$ ($\tau=\sum_v \mu_v$, $\delta^2=\sum_v\mu_v(1-\mu_v)$). Further assuming the observation time window $T$ is weighted exponentially with parameter $\lambda$, we have the following theorem. 



\begin{theorem}
\label{theorem:powerlaw}
$P(N_z)=C N_z^{\alpha}$, where $C=\frac{\lambda}{\delta\sqrt{(\frac{\tau}{\delta})^2+2\lambda}}$ and $\alpha=-1+\frac{\tau}{\delta^2}-\frac{\sqrt{(\tau^2+2\lambda \delta^2)}}{\delta^2}$. 
\end{theorem}

\begin{proof}
For action $z$, $P(N_z)$ is equivalent to the mixture of group models whose observation time parameter $T$ is weighted exponentially. Formally, we have

\beq{
\label{eq:pnd}
\small
\besp{
P(N_z)=\int_0^{\infty} \lambda \exp(\lambda t) \frac{1}{N_z\delta \sqrt{2\pi t}} \exp(\frac{-(ln{N_z-\tau t})^2}{2\delta^2 t})dt.
}
\normalsize
}

Let $a=\lambda - \frac{\tau^2}{2\delta^2}$, $b=-\frac{(\ln n_d)}{2\delta^2}$, and $c=\frac{\lambda}{n_d\sqrt{2\pi \delta^2}\exp(\frac{\tau\ln n_d}{\delta^2})}$. It is obvious that $a, b \leq 0$. Then we have 

\beq{
\label{eq:pnd_2}
\besp{
Eq.~\ref{eq:pnd}&= c\int_0^{\infty} \frac{1}{\sqrt{t}}\exp(at+\frac{b}{t})dt \\
&= 2c \int_0^{\infty} \exp(at^2+\frac{b}{t^2})dt^2.
}
}

Let $s=t^2\sqrt{\frac{b}{a}\cdot \frac{1}{t^2}}$, we have

\beq{
\label{eq:pnd_3}
\besp{
Eq.~\ref{eq:pnd_2}&=c\exp(2\sqrt{ab})\int_{2(\frac{a}{b})^{\frac{1}{4}}}^{\infty} \frac{\exp(as^2)}{\sqrt{s^2-4\sqrt{\frac{b}{a}}}}ds^2.
}
}

Let $u=s^2-4\frac{b}{a}$, we have

\beq{
\small
\besp{
\label{eq:pnd_4}
Eq.~\ref{eq:pnd_3}&=c\exp(2\sqrt{ab})\int_0^{\infty}\frac{\exp(au-4\sqrt{ab})}{\sqrt{u}}du\\
&=\frac{c\exp(-2\sqrt{ab})}{\sqrt{-a}}\Gamma(\frac{1}{2})\\
&=\frac{\lambda}{N_z\sqrt{2\delta^2}\sqrt{\frac{\tau^2}{2\delta^2}-\lambda}}\exp(\frac{\tau}{\delta^2}\ln N_z-2\frac{\ln N_z}{\delta}\sqrt{\frac{1}{2}(\frac{\tau^2}{2\delta^2}-\lambda)})\\
&=\frac{\lambda}{\sqrt{(\tau^2-2\delta^2\lambda)}}N_z^{-1+\frac{\tau}{\delta^2}-\frac{\tau^2-2\lambda\delta^2}{\delta^2}}
}\normalsize
}

Hence, we have $P(N_z)=C N_z^{\alpha}$, where


\beqn{
\label{eq:meso2macro}
\besp{
C &=\frac{\lambda}{\sqrt{(\tau^2-2\delta^2\lambda)}}, \\
\alpha &=-1+\frac{\tau}{\delta^2}-\frac{\tau^2-2\lambda\delta^2}{\delta^2}.
}
} 
\end{proof}

\hide{
\begin{proof}
See details in Appendix. 
\end{proof}
}

For $\alpha < 0$, we say $N_z$ is power-law distributed. 
\hide{
For an action with unknown collective model parameters, the probability of $N_z$ is given by:

\beq{
\label{eq:pn_hidden}
P(N_z)=\sum_{i} P(N_z|\tau_i, \delta_i)P(\tau_i, \delta_i)
}

For instance, under the assumption that each group model's parameters $\Theta_{\text{meso}}=\{(\tau_i, \delta_i)\}$ distributes uniformly, Eq.~\ref{eq:pn_hidden} can be rewritten as

\beq{
\label{eq:pn_all}
P(N_z)=\frac{1}{|\Theta_{\text{group}}|}\sum_z \frac{C_i}{N_z^{\alpha_z}},
} 

\noindent which behaves like a power-law with an exponent given by the smallest power $\alpha_z$.  This leads to the general result that the evolutionary of $N_z$ distributed according to a power-law.
}
It is worth noting that, when group behavior follows a lognormal distribution, 
even without the conditions of $\phi=\tau t $ or $\varphi = \delta^2 t$, 
 the power-law result still holds. 
 The proof can be obtained by extending the proof for Theorem~\ref{theorem:powerlaw}. A similar study was conducted in~\cite{adamic1999nature}. 

\vpara{Behavior of \model: when does the ``winner take all''?}
Let us now examine the behavior of our framework to see when the winner-take-all mechanism holds and leads to the power-law phenomena. Consider a simple system, in which each user shares the same individual model with the parameter $p$ to perform an action. It turns out there is a lower bound on $p$ for the existence of the power-law distribution over $N_z$. 


\begin{theorem}
$N_z$ is power-law distributed when $p > \frac{1}{m}$, where $m$ is the number of users in $G$. 
\end{theorem}

\begin{proof}
According to Eq.~\ref{eq:micro2meso}, we have $\tau=mp$ and $\delta=mp(1-p)$. Also, according to Eq.~\ref{eq:meso2macro}, we obtain that $\beta=-1+\frac{1-mp}{1-p}+2\lambda$. 

The power-law holds for $N_z$ when $\alpha < 0$, that is, 

\beqn{
\label{eq:power_law_hold}
\besp{
\alpha < 0 
\Rightarrow& -1+\frac{1-mp}{1-p}+2\lambda < 0 \\
\Rightarrow& p > \frac{2\lambda+2}{m+1} > \frac{1}{m}.
}
}
\end{proof}

When does the ``winner take all''? Assuming we have a purely random ($p=0.5$) system for the evolution of group behavior, as $m$ is usually large, we are safe to claim that a power-law holds. However, considering a deterministic system, in which no user will adopt any action ($p \rightarrow 0$), the power-law will fail. 


Through the above discussions, we are given some insight into the fact that randomness and the aggregate effect of individual actions finally result in 
the macroscopic-power-laws phenomena.  


\vpara{Significance.} 
We provide evidence of how heavy-tailed phenomena in social networks emerge from individual actions. 
Specifically, a power-law distribution is determined by the parameters derived from our individual models and group models. 
We provide theoretical conditions under which a power-law distribution can form in social networks.

\subsection{Further Discussions}

\vpara{Application.} We discuss some potential applications of the proposed framework. One can integrate any machine learning algorithm into the individual model. For instance, when studying the information diffusion process in Twitter,
 for modeling user $v$'s retweeting behavior,
  we can define features (e.g., the likelihood of the user's interests matching with the tweet's topics, profiles of the user, etc.) and construct a feature vector $\mathbf{s_{vz}}$. We then use a Logistic function $f(\cdot)$ with $s_{vz}$ to represent the individual model $\mu_{vz}$. With efficient training samples, we are allowed to estimate $f(\cdot)$'s parameters using Maximum Likelihood Estimation (MLE)~\cite{akaike1998information}. Notice that the logistic function can be replaced by any other classification or regression models.
   After that, we are able to generate the ``adopt'' decisions of individual users, for the instantiation of the group model according to Eq.~\ref{eq:meso}, and demonstrate how the retweets each tweet receives evolve over time. At last, we can calculate the parameters of the heavy-tailed distributions to present the macroscopic phenomena. 

\vpara{Connection with previous work.} 
The proposed framework can be viewed as a generalization of several existing models. 
In Eq~\ref{eq:meso}, when $U=D$, the group model can be viewed as a generalized Black-Scholes option pricing model~\cite{black1973pricing}. The connection with individual models and group models is a natural multiplicative process~\cite{mitzenmacher2004brief}. When group behaviors follow lognormal distributions, while the conditions of $\phi=\tau t $ or $\varphi = \delta^2 t$ are not satisfied, the integration of group behaviors with heavy-tailed phenomena is similar to that of the evolutionary process of sites on the Web~\cite{adamic1999nature}.

\section{Experimental Setup}
\label{sec:observation}


\begin{table}[t]
\centering \caption{\label{tb:data} Statistics of the datasets.}
\small
\renewcommand\arraystretch{1.2}
\begin{tabular}{c|c|c|c}
\hline & Weibo & Citation & Flickr \\
\hline Users & 326,497,021 & 1,712,431 & 259,565 \\
\hline Posts/Papers/Photos & 4,634,168,136 & 2,092,356 &  854,734 \\
\hline Users' relations & 3,274,895,719 & 6,485,521 &  1,898,069\\
\hline Action logs & 1,026,243,542 & 8,012,227 & 3,884,739 \\
\hline Time period & 4 months & 79 years & 12 months \\
\hline
\end{tabular}
\normalsize
\end{table}

\begin{figure*}[t]
\centering
\mbox{
\subfigure[3rd hour on Weibo]{
\label{fig:group_weibo}
\epsfig{file=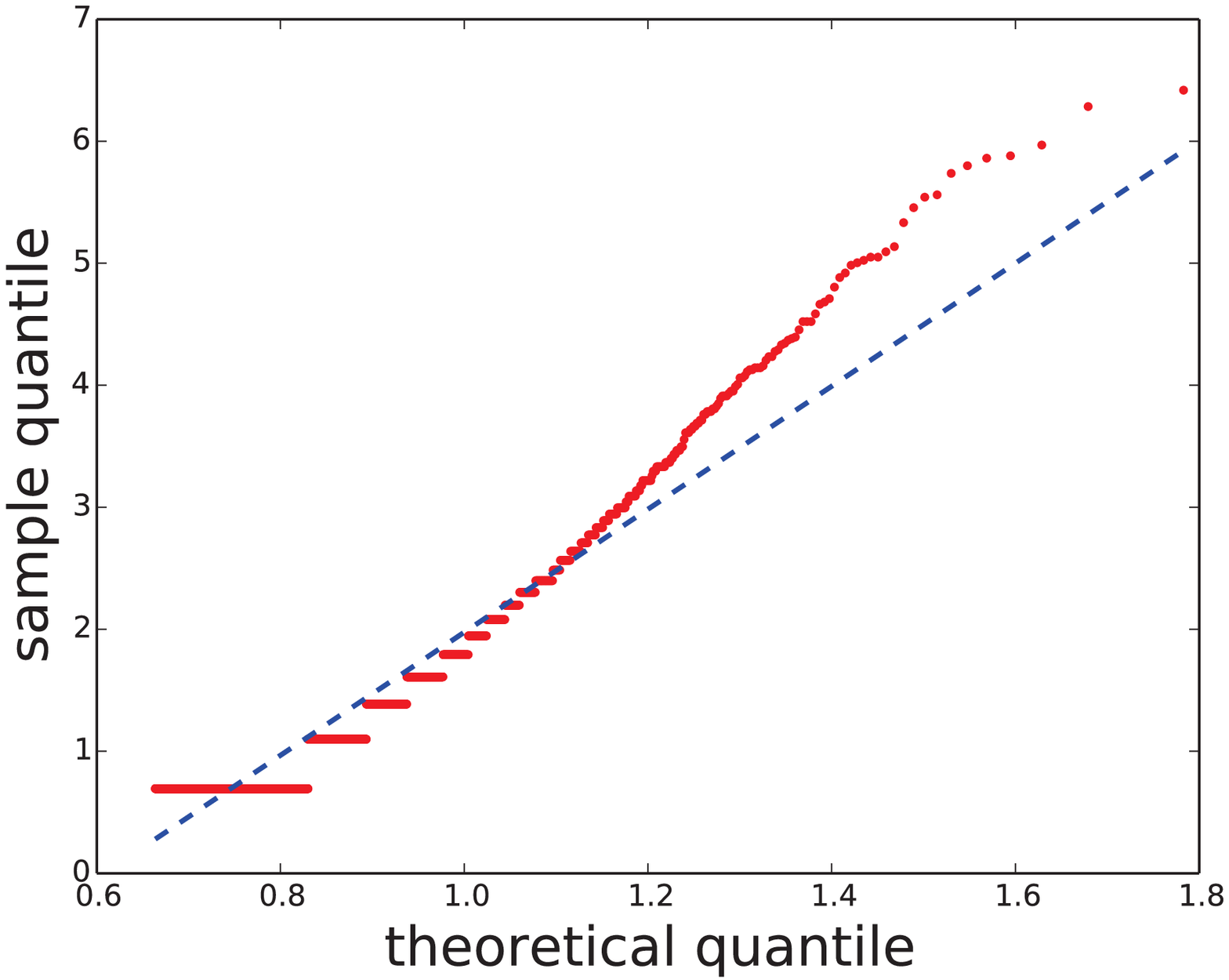, width=2in}
}
\subfigure[9th year on Citation]{
\label{fig:group_citation}
\epsfig{file=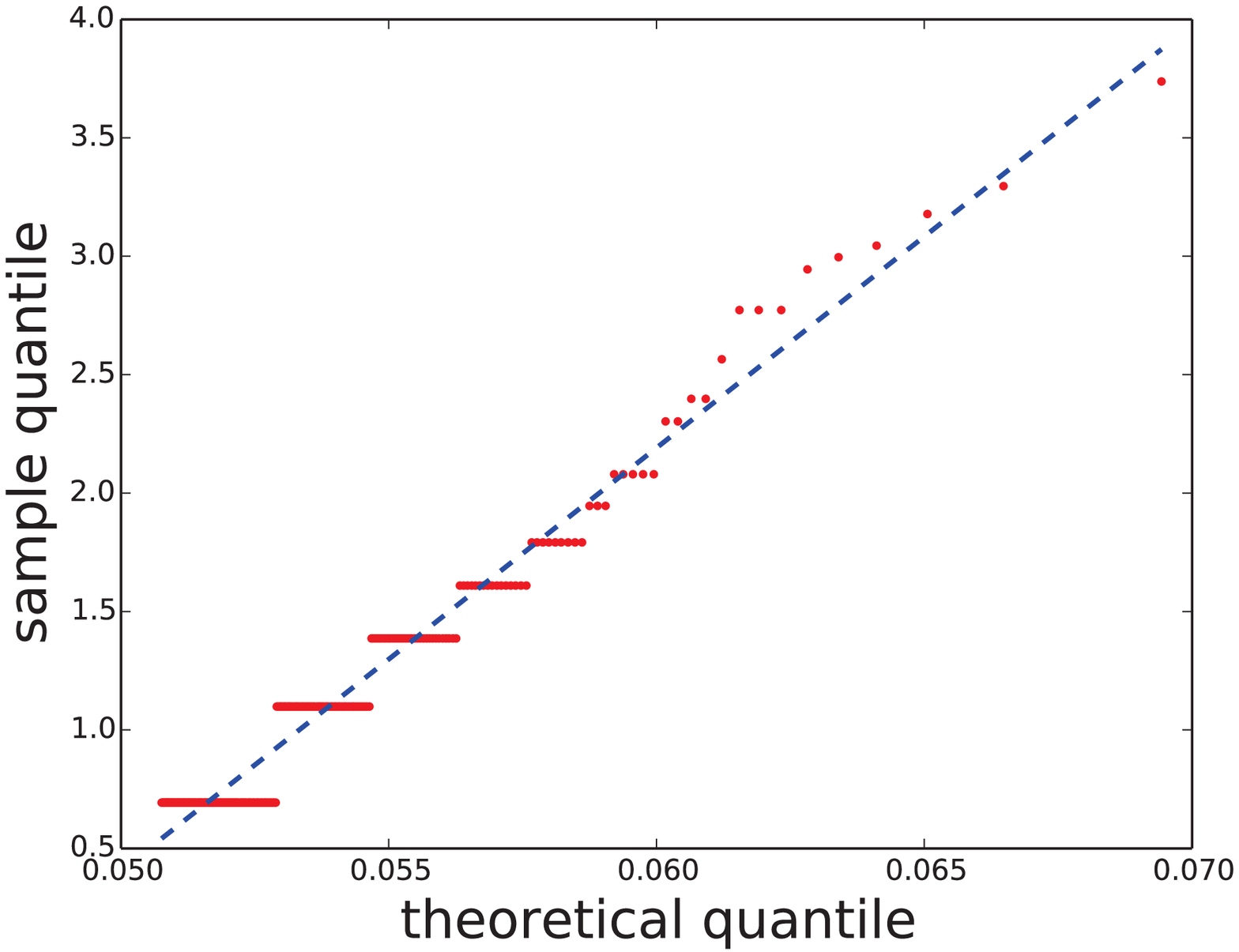, width=2in}
}
\subfigure[23rd hour on Flickr]{
\label{fig:group_photo}
\epsfig{file=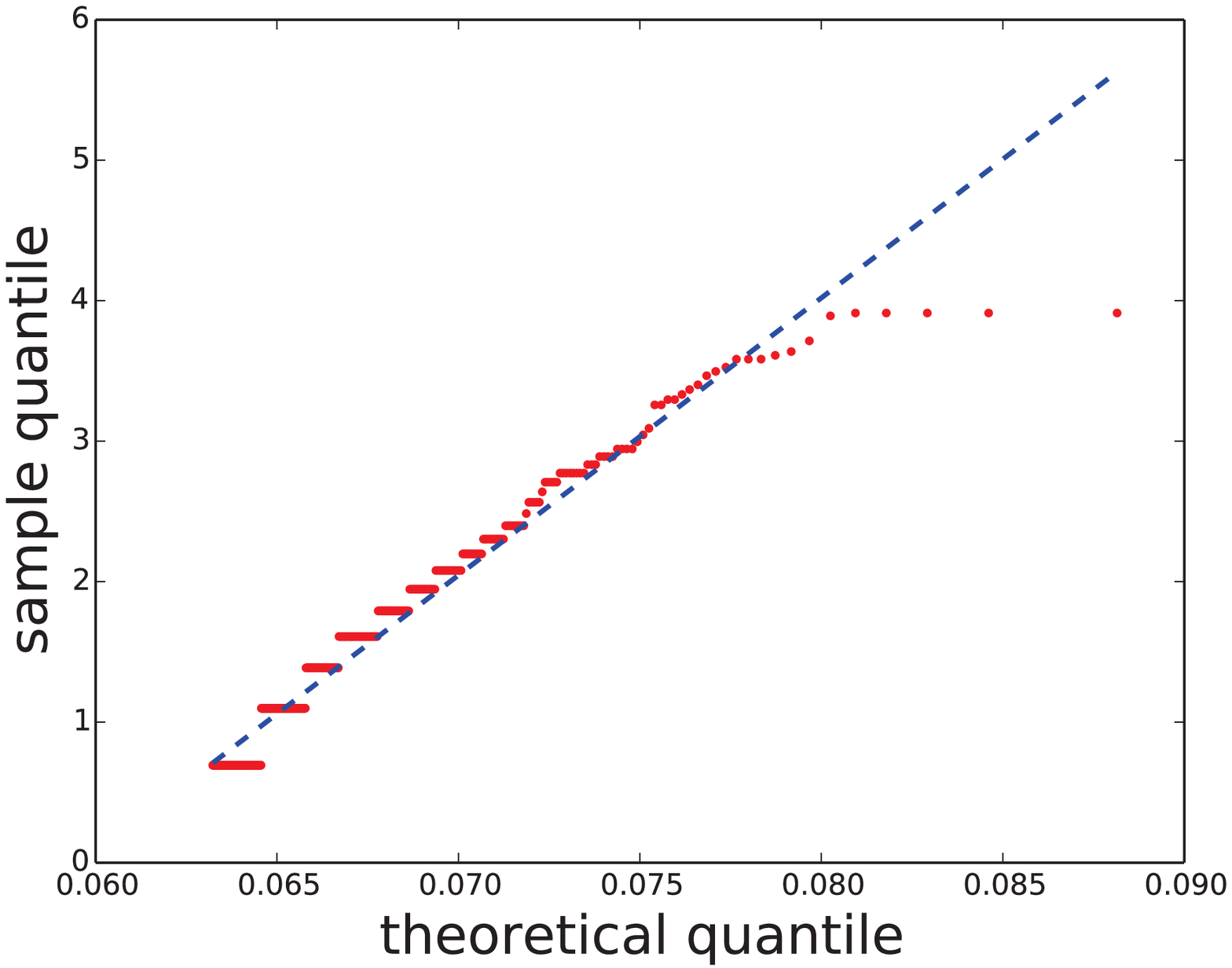, width=2in}
}
}
\caption{QQ-plots for analyzing the distribution of group behavior at a particular timestamp on each dataset.} 
\label{fig:group_behavior}
\end{figure*}

\begin{figure*}[t]
\centering
\mbox{
\subfigure[]{
\label{fig:macro_post}
\epsfig{file=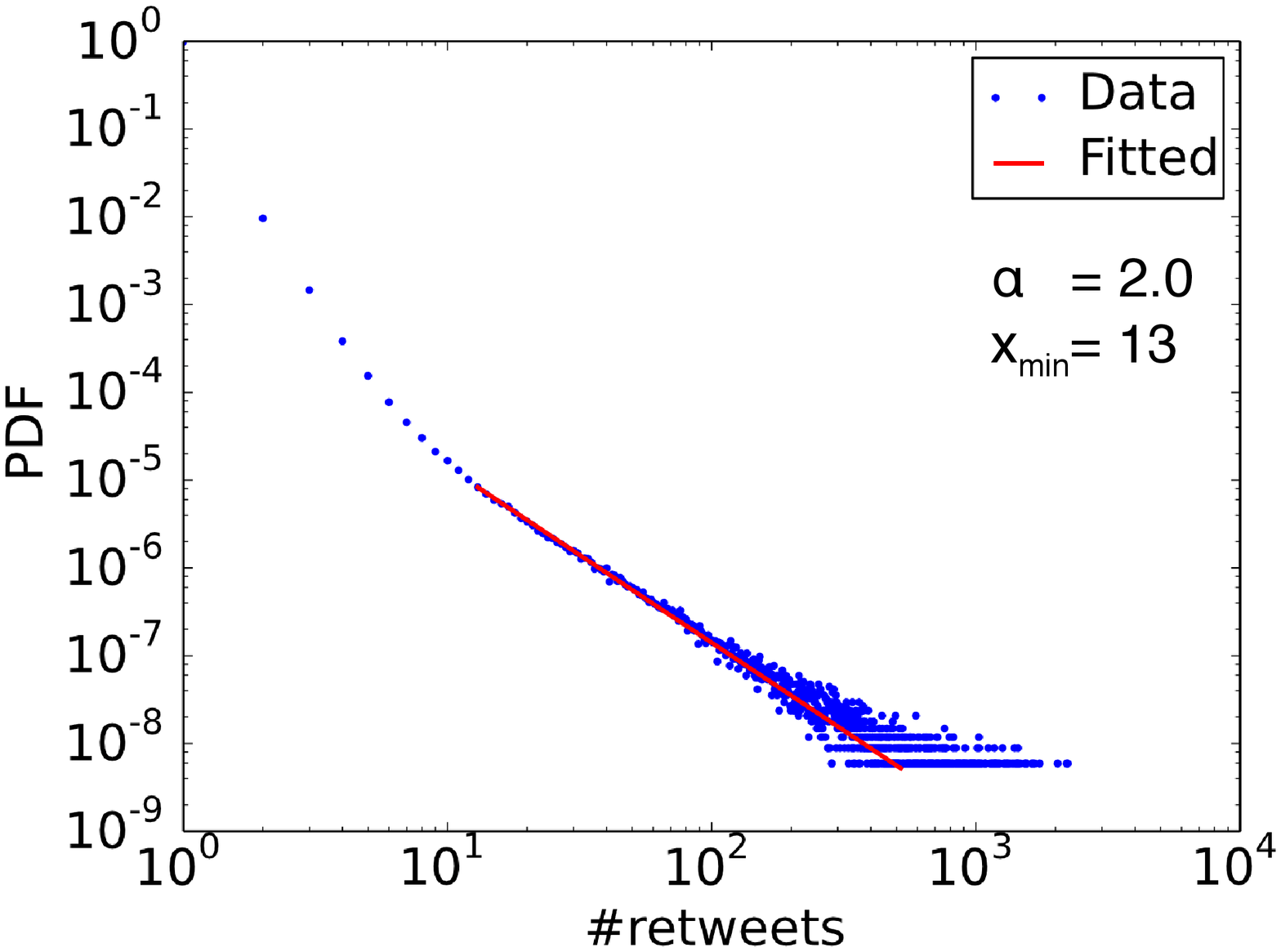, width=2in}
}
\subfigure[]{
\label{fig:macro_paper}
\epsfig{file=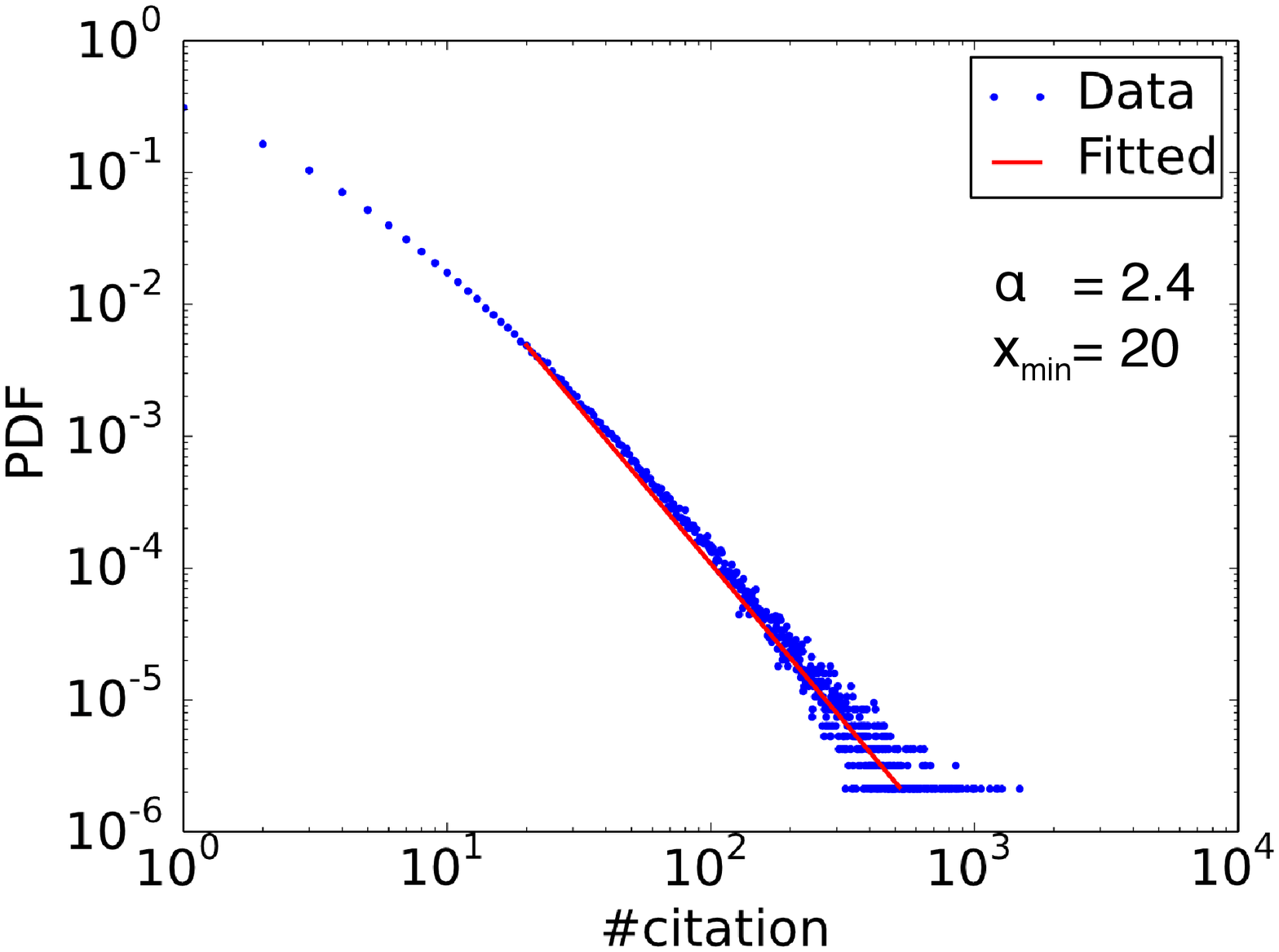, width=2in}
}
\subfigure[]{
\label{fig:macro_flickr}
\epsfig{file=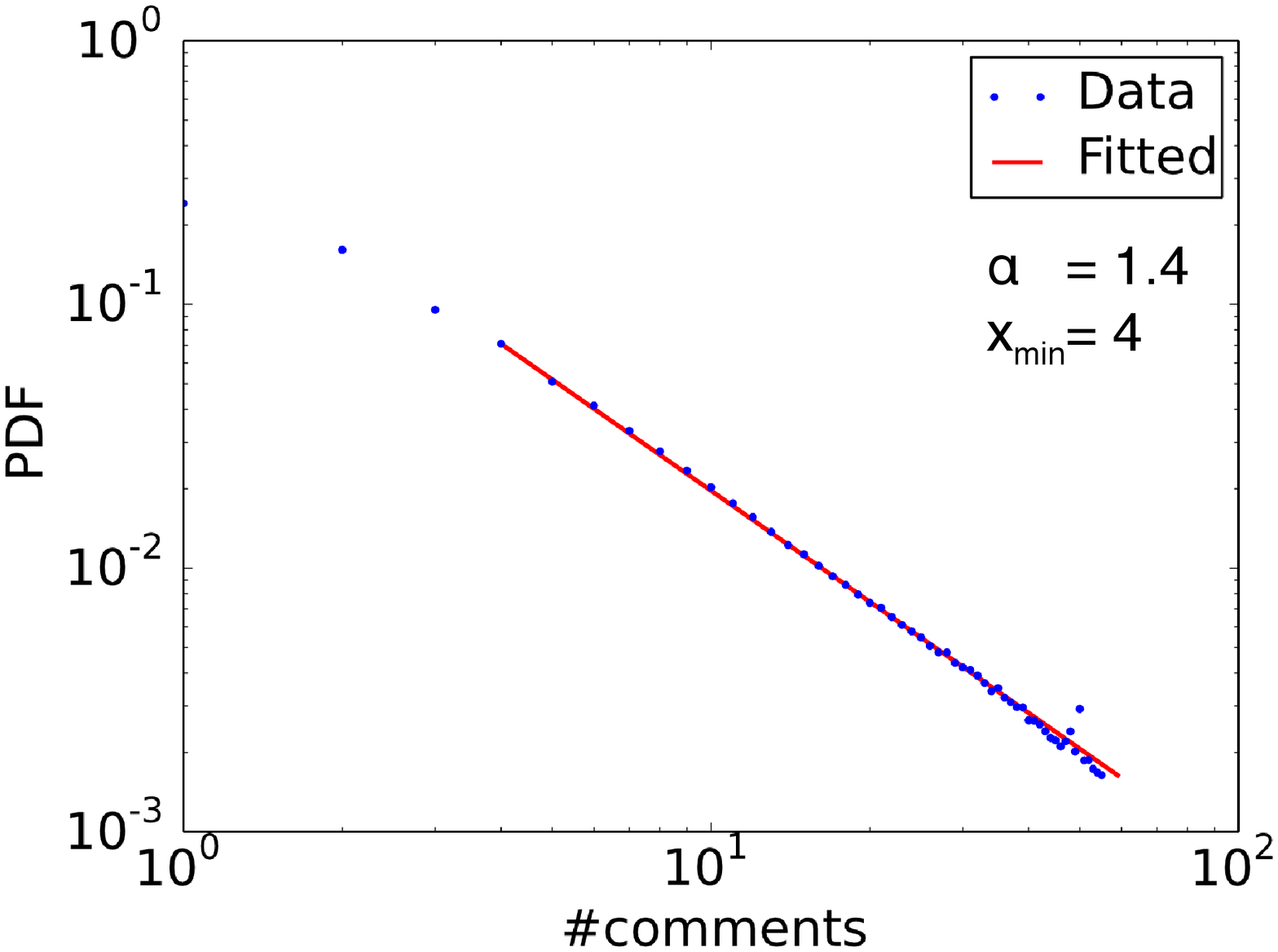, width=2in}
}
}
\caption{Heavy-Tailed Phenomena in Real-World Networks. Blue dots are observed data. Red lines are power-law distributions which best fit the data, with estimated exponent parameters and truncation points represented in the figures. All figures are plotted on a log-log scale. } 
\label{fig:macro}
\end{figure*}

\subsection{Datasets}

We verify the proposed framework on three different genres of large datasets:
Tencent Weibo, Citation, and Flickr,
Statistics of the three datasets are summarized in Table~\ref{tb:data}. 

\vpara{Tencent Weibo~\cite{yang2014rain}.} It is one of the most popular microblogging services in China. 
The dataset consists of 326,497,021 users, 3,274,895,719 following relationships, and 4,634,168,136 tweets, spanning over 4 months between Oct. 1st, 2011 and Jan. 30th, 2012. 
In this dataset, we define the user's retweeting behavior as the individual action. 

\vpara{Citation~\cite{Tang:08KDD}.} It is from ArnetMiner\footnote{http://aminer.org/citation}.
This  dataset consists of 2,092,356 papers published during 1950 and 2012. From those papers, we derive 1,712,433 authors, and 8,012,227 citation relationships  between papers.
In this  dataset, the individual action is defined by a paper's citation behavior (determined by the authors in the course of scientific writing). 

\vpara{Flickr~\cite{yang2014your}.} It was crawled from Flickr. 
This dataset contains 854,734 photos and 3,884,739 comments generated by 259,565 users. 
The dataset is used to investigate user commenting actions in photo sharing networks. 
Specifically, we define the individual action as whether a user posts a new comment to a specific photo.

\subsection{Data Analysis}

\vpara{Group behavior.} 
We examine how group behavior $n_z(t)$, such as $\#$retweets, $\#$citations, and $\#$comments at each timestamp $t$, on real datasets distributes. 

To test the hypothesis that the group behavior at a particular timestamp is lognormally distributed, we utilize QQ-plot~\cite{Wilk:QQ1968}, which is commonly used to compare two distributions by plotting their quantiles against each other. Specifically, for timestamp $t$, we first estimate the parameters of empirical normal distribution of $\ln n_z(t)$, by a MLE method, which is also used in~\cite{zaman2014bayesian}. We then plot the quantiles of $\ln n_z(t)$ in real data against certain quantiles of the estimated normal distribution. 

Figure~\ref{fig:group_behavior} shows the results on all datasets. 
For instance, in Figure~\ref{fig:group_weibo}, we test the distribution form of $\#$retweets that a tweet receives at the 3rd hour since that tweet is posted in Weibo network. 
The approximate linearity of the plotted points suggests $\#$retweets, at the 3rd hour since the original tweet is posted, follows a lognormal distribution. Analogously, we observe similar results on $\#$citations and $\#$comments in the Citation and Flickr networks, as shown in Figures~\ref{fig:group_citation} and~(c). We also archieve similar results at other timestamps on all datasets. 

\hide{
We test whether the collections of individual actions at each timestamp in networks, such as $\#$retweets, $\#$citations, and $\#$comments, follow the lognormal distributions as our framework suggests. 
In statistics, a QQ plot is commonly used to compare two probability distributions by plotting their quantiles against each other~\cite{Wilk:QQ1968}. 


Figure~\ref{fig:group_behavior} shows the QQ plots for comparing empirical distributions of group behaviors in three networks with theoretical lognormal distributions, estimated using Maximum Likelihood Estimation (MLE)~\cite{zaman2014bayesian}. 
The $x$-axis represents the theoretical quantile and $y$-axis represents the data quantile. 
The points in the QQ figure will approximately lie on the linear line \hide{$y=x$} if two distributions are similar. 
In Figure~\ref{fig:group_weibo}, we aim to test the distribution form of $\#$retweets that a tweet receives at a particular timestamp in the Weibo network. 
The approximate linearity of the plotted points suggests the collection of retweets each post receives at each timestamp is lognormally distributed. 
Analogously, we observe similar results on citing and commenting actions in the Citation and Flickr networks, as shown in Figures~\ref{fig:group_citation} and~(c). 
This indicates that the proposed framework is able to correctly capture the distributions of group behaviors in social networks. 
}

\vpara{Heavy-Tailed distribution.} 
We now examine the network distribution, $P(N_z)$, on real datasets. To do so, in Weibo (Citation or Flickr) network we plot $\#$retweets ($\#$citations or $\#$comments) each tweet (paper or photo) receives within 4 months (79 years or 12 months) in Figure~\ref{fig:macro}. 
The linearity on log-log scales suggests a heavy-tailed distribution on all datasets. 

We try to fit power-law distributions on the three datasets by a classical fitting technique~\cite{clauset2009power} which determines two parameters: a truncation point $x_{min}$ that governs the lower bound, above which the observed data obey a power-law distribution, and the exponent parameter $\alpha$ of the potential power-law distribution. 

We present the power-law distributions that best fit the data in each network in Figure~\ref{fig:macro} by blue lines. We observe that the power-law distributions in the Weibo network with exponent $2.0$, the Citation network with exponent $2.37$, and the Flickr network with exponent $1.4$. 
The corresponding truncation points $x_{min}$ in each network are 13, 20, and 4, respectively. 
We use the $p$-value~\cite{clauset2009power} and Residual Sum of Squares (RSS)~\cite{draper2014applied} together to test the power-law hypothesis.  
We find that the $p$-values are all greater than 0.1. The average value of RSS on each dataset is $2.3\times 10^{-6}$ (see details in Table~\ref{tb:rss}). The $p$-value, together with the small valued RSS scores, suggest that the real-world data exhibits power-law distributions~\cite{clauset2009power}. 
We also try to fit lognormal distributions to the observed data. However, we achieve negative results on all datasets. 

\hide{
We now test 
whether the integration of individual actions over observed timespan results in a power-law distribution as suggested by \model.  
To do so, we fit power-law models to the three datasets by using the classical fitting technique~\cite{clauset2009power}. 
This method is determined by two parameters: one is a truncation point $x_{min}$ that governs the lower bound, above which the observed data obey a power-law distribution, and the other one is the exponent parameter $\alpha$ of the potential power-law distribution. 

Figure~\ref{fig:macro} plots the original distributions from real data and also the power-law distributions that best fit the data in each network. 
We observe that the power-law distributions in the Weibo network with exponent as $2.0$, the Citation network with exponent as $2.37$, and the Flickr network with exponent as $1.4$. 
The corresponding truncation points $x_{min}$ in each network are 13, 20, and 4, respectively. 
We use the $p$-score~\cite{clauset2009power} and Residual Sum of Squares (RSS)~\cite{draper2014applied} together to examine the fitting results. 
We find the $p$-scores of the fitting models are all greater than 0.1, and the average value of RSS on each dataset is $2.3\times 10^{-6}$ (see details in Table~\ref{tb:rss}). The $p$-scores, together with the small valued RSS scores, suggest that the real-world data exhibits power-law distributions~\cite{clauset2009power} as our framework assumes. 
}

\subsection{Evaluation Measures}
\label{sec:setting}

It is difficult to find a ground truth to evaluate the accuracy of the proposed framework. 
As the most important feature of \model\ is the connection of macro-level network distribution with meso-level group behavior and micro-level individual actions; for quantitatively evaluation, we apply \model\ to fit the (macro-level) heavy-tailed distribution from observed individual actions. 
One advantage of \model\ is that it can fit a network distribution from partially observed data. The other advantage is that we can use the group behavior or individual actions to explain the formation of the fitted distributions.

Another idea to evaluate \model\ is to apply it to some prediction task.  \model\ can better capture group behavior in social networks. Thus we apply \model\ to information burst prediction and evaluate the prediction performance in terms of Precision, Recall, and F1-Measure.


\begin{figure*}[t]
\centering
\mbox{
\subfigure[]{
\label{fig:macro_pred_tweet}
\epsfig{file=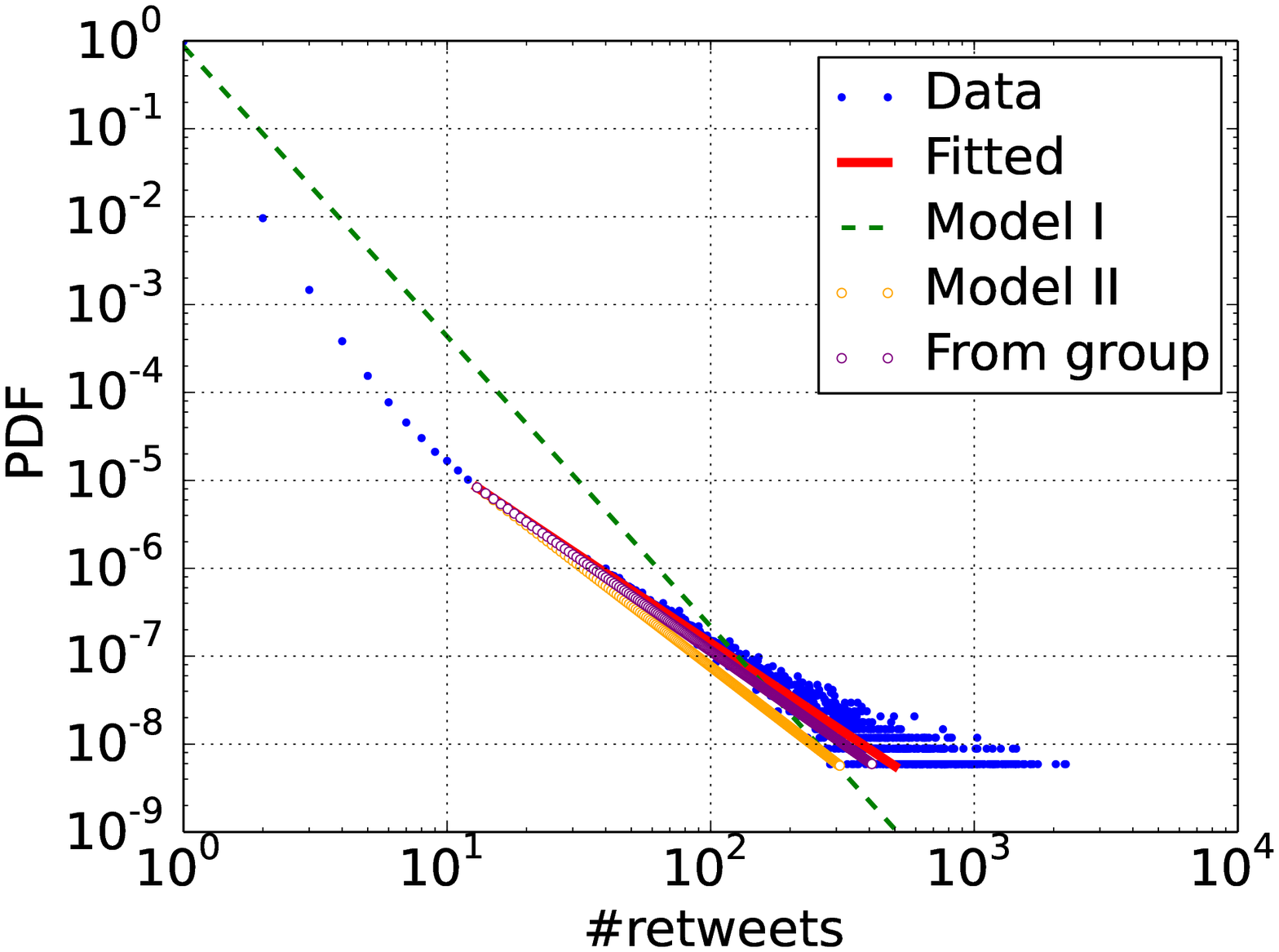, width=2.2in}
}
}
\hspace{-0.2in}
\mbox{
\subfigure[]{
\label{fig:macro_pred_paper}
\epsfig{file=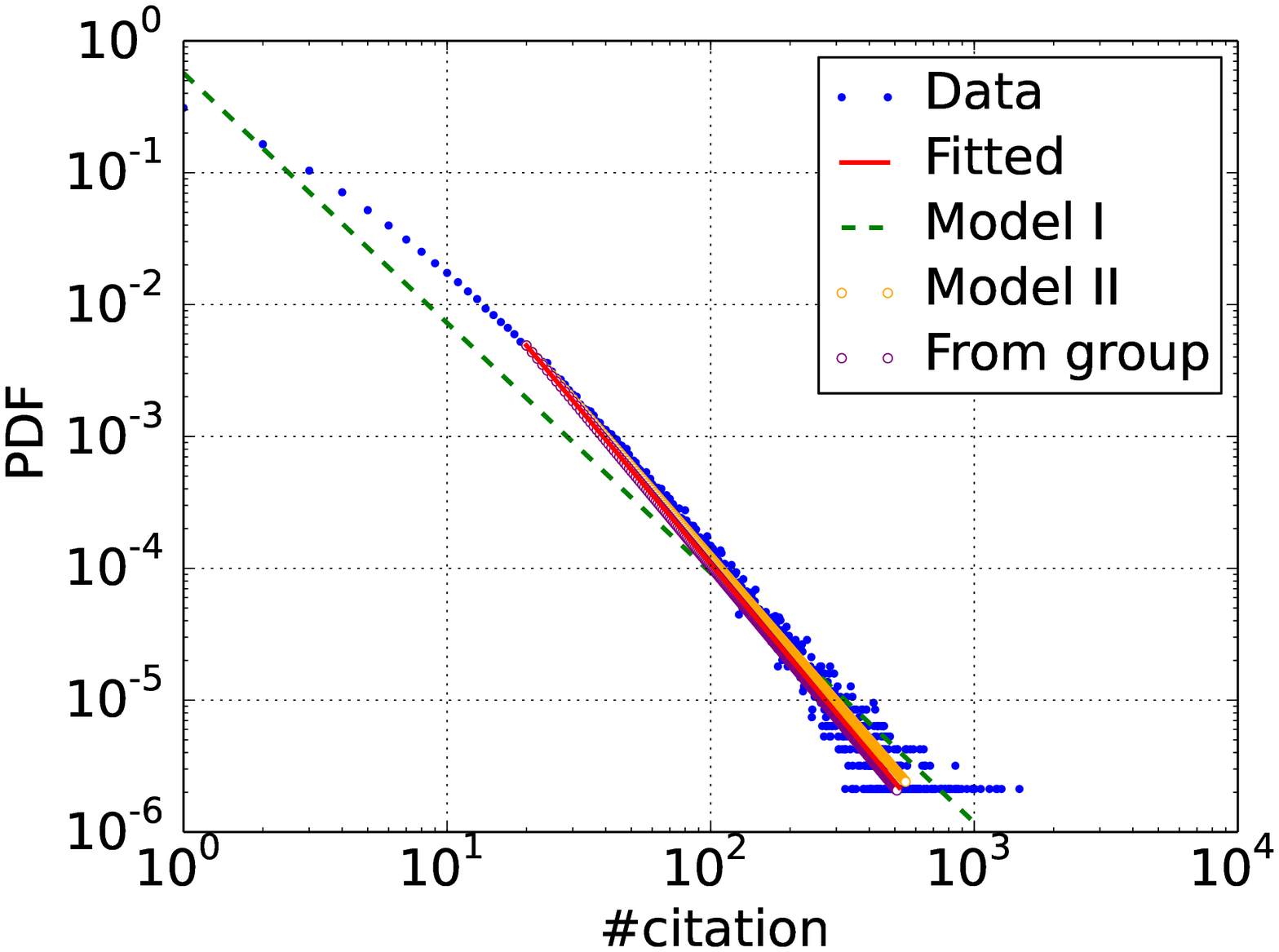, width=2.2in}
}
}
\hspace{-0.2in}
\mbox{
\subfigure[]{
\label{fig:macro_pred_flickr}
\epsfig{file=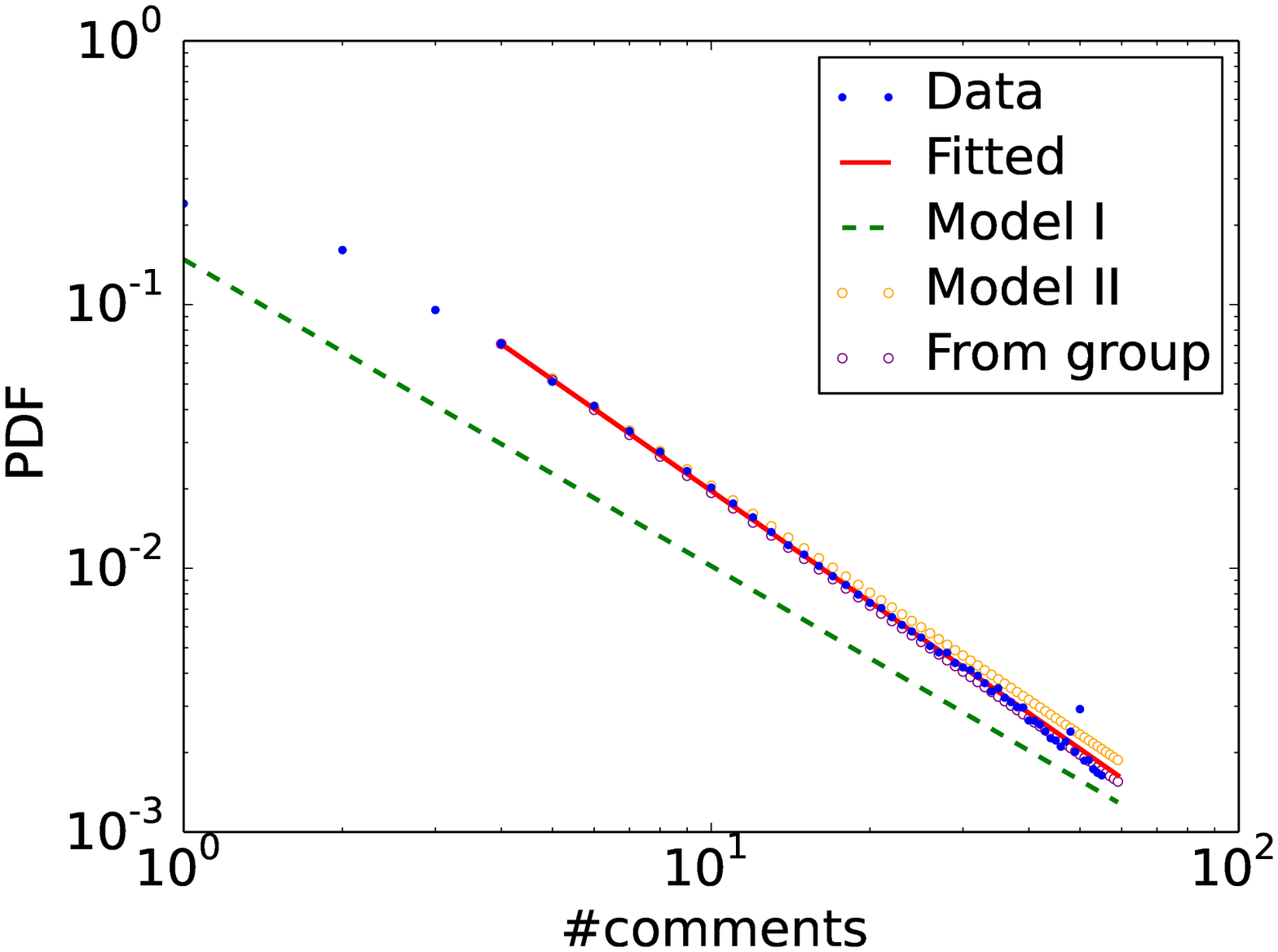, width=2.2in}
}
}
\caption{Fitting the heavy-tailed phenomena in (a) Weibo, (b) Citation, and (c) Flickr. All figures are plotted on a log-log scale. In each figure, x-axis denotes the number of retweets (citations or comments). Y-axis denotes the fraction of tweets (papers or photos). Blue solid dots are observed data. The red solid line is the result of fitting the observed data. The green dashed line is the result obtained by our framework, using partially observed individual actions. Yellow hollow dots are obtained by our framework considering truncation points. Purple hollow dots represent the fitting results using group behaviors. } 
\label{fig:macro_pred}
\end{figure*}

\section{Experimental Results}
\label{sec:exp}

To quantitatively validate \model, 
we consider the following three evaluation aspects:

\begin{itemize}

\item \textbf{Fitting heavy-tailed phenomena using partially observed actions:} We thus examine to what  extent \model \ can capture the emerging process of heavy-tailed phenomena by using partially observed individual actions in real social networks. 

\item \textbf{Group behavior prediction:} By this we examine the extent to which \model\ can model the aggregate effect of group behaviors from individual actions in real social networks. 

\item \textbf{Information burst prediction:} Finally, we use this application to further demonstrate the effectiveness of \model.
\end{itemize}



\begin{table}[t]
\centering \caption{\label{tb:rss} Residual Sum of Squares (RSS) for fitting heavy-tailed phenomena. }
\small
\begin{tabular}{c|c|c|c}
\hline & Weibo & Citation & Flickr \\
\hline Fitted & $6.180 \times 10^{-13}$ & $1.426 \times 10^{-6}$ & $5.606\times 10^{-6}$ \\
\hline Model I & $8.240\times 10^{-1}$ & $7.887\times 10^{-2}$ & $3.867 \times 10^{-2}$ \\
\hline Model II & $1.268\times 10^{-10}$ & $7.366\times 10^{-7}$ & $7.000 \times 10^{-4}$ \\
\hline From Group & $8.062\times 10^{-13}$ & $1.721\times 10^{-6}$ &  $1.059\times 10^{-5}$\\
\hline
\end{tabular}
\normalsize
\end{table}

\subsection{Fitting Heavy-Tailed Phenomena using Partially Observed Actions}
\label{sec:exp_macro}

This task is to demonstrate whether \model \ can capture the emerging process of heavy-tailed phenomena by using only partially observed individual actions in real social networks. 


\vpara{Problem.}
Given an observed subnetwork $G=(V, E)$ of a complete network $H$ over a time window $T$ and a set of individual actions $X=\{x_{tvz} | t \leq T, v\in V\}$ in $G$, the goal is to estimate whether $N_z$ in real networks---the number of users who perform the action $z$ in the complete network $H$ over the observed timespan $T$---follows the power-law distribution parameterized by our framework. 

Specifically, in the Weibo network, $x_{tvz} = 1$ indicates that user $v$ retweets a particular tweet, indicated by $z$, at timestamp $t$. In Citation, it means author $v$ cites another a particular paper $z$ at time $t$. 
Analogously, the action that user $v$ posts a comment to photo $z$ at time $t$ is denoted by $x_{tvz} = 1$ in the Flickr network. 


\vpara{Setup.}
In Section~\ref{sec:observation}, we conclude that all three networks exhibit power-law distributions and also provide the estimated exponent parameters that best fit the real data. 
We use Residual Sum of Squares (RSS)~\cite{draper2014applied} to quantify the distance between the distributions of $N_z$ in real data and the distributions provided by our framework. A smaller RSS represents a better-fitted distribution.

We introduce how we generate the observed network $G$ in different datasets. 
In the Weibo network, we choose the users who retweet a tweet $z\in Z$, where $Z$ is the set of all tweets exposed to users $V$ in $G$, within the first 50 minutes since $z$ was posted, and the followers/followees of these users. 
In Citation, $G$ consists of authors who cite papers within the first year since these papers are published. 
In Flickr, the users who posted comments to photos within the first hour since the photos are posted are chosen as the subnetwork $G$.

We introduce two methods to apply our framework to estimate the exponent parameter $\alpha$ of power-law distributions from individual actions in real data.  

\textit{Model I.} We model individual actions by a binomial model with parameter $\mu_{vz}$ in Definition~\ref{def:individual_model}.  
There are different ways to represent and estimate $\mu_{vz}$, such as using a logistic regression to represent $\mu_{vz}$ and estimating the regression parameters by the Maximum Likelihood Estimation (MLE) method according to the individual action logs of user $v$. 
In this work, to keep our framework flexible and general, we use a straightforward method to represent $\mu_{vz}$ --- that is, the probability that user $v$ is influenced by one of her neighbors to perform the action $z$.  Specifically, we define $\mu_{vz}$ as


\beq{
\label{eq:mu_app}
\mu_{v} = \frac{\sum_t x_{tvz}}{\sum_t \sum_{u|e_{vu}\in E} x_{tuz}}
}

Given $mu_{vz}$, we then calculate $\tau$ and $\delta$---the parameters of group behaviors in Eq.~\ref{eq:micro2meso}. Unfortunately, there is no effective way to estimate $\lambda$ (the exponential parameter of the observation time window $T$) automatically. Thus we define $\lambda$ manually and leave the automatic estimation method to our future work. In practice, we empirically set $\lambda$ as $0.008$, $0.1$, and $0.23$ in Weibo, Citation, and Flickr, respectively. We finally obtain the exponent $\alpha$ according to Eq.~\ref{eq:meso2macro}. 

\textit{Model II.} As prior work~\cite{clauset2009power} and Figure~\ref{fig:macro} suggest, however, few empirical phenomena in practice obey power-laws for all observed data $x$. More often the power-law applies for values greater than some minimum point $x_{\text{min}}$, which can be understood as the truncation point of the empirical phenomena. 
We first calculate the truncation point $x_{\text{min}}$ by the MLE method described in ~\cite{clauset2009power}. 
We then use the posts whose retweeting numbers are no less than $x_{\text{min}}$ to estimate  parameters $\tau$ and $\delta$ by following the steps in Model I. 
After that, when normalizing the estimated PDF, we also only consider the tweets whose retweets are no less than the lower bound. 


\begin{figure*}[t]
\centering
\mbox{
\subfigure[]{
\label{fig:pred_group_post}
\epsfig{file=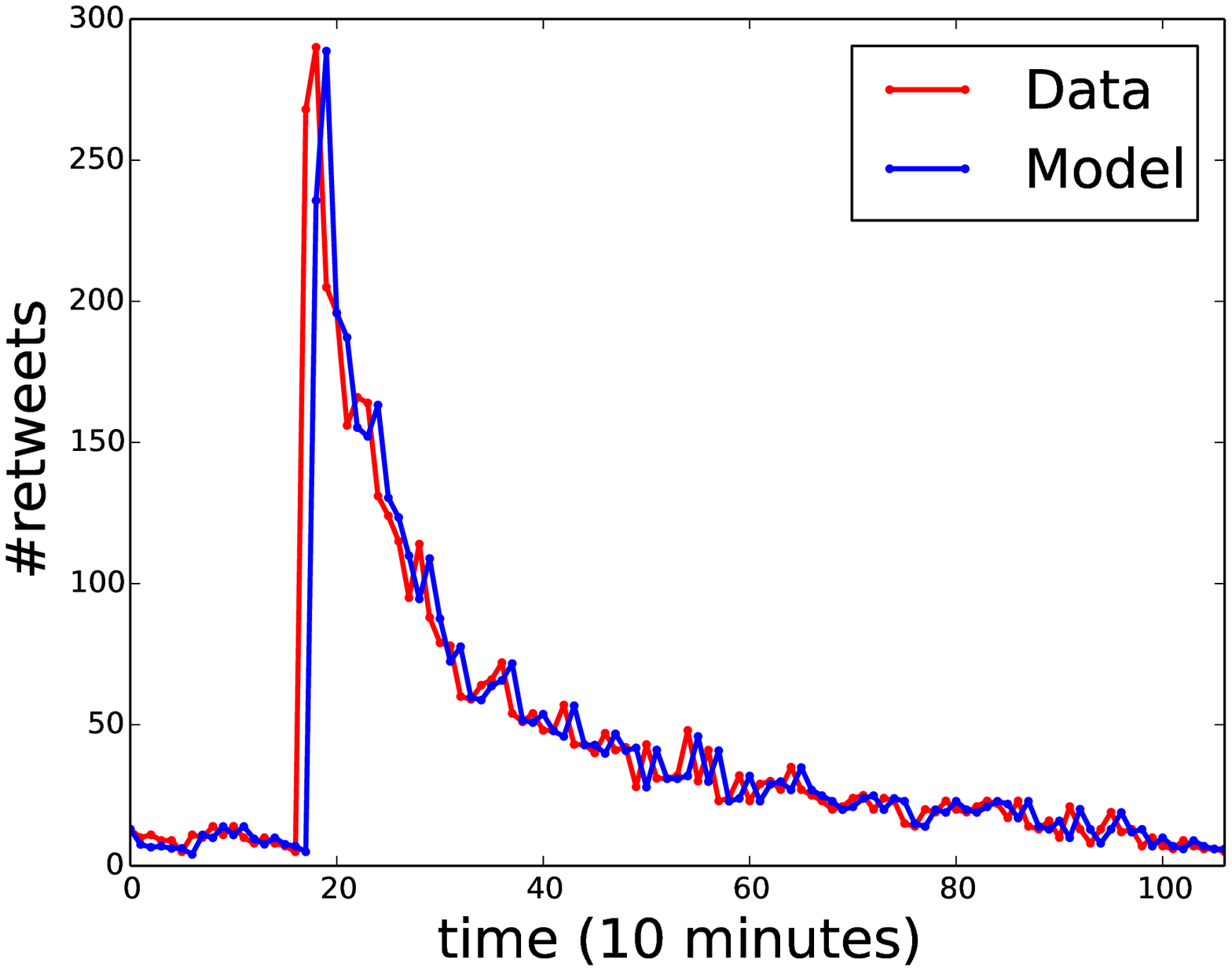, width=2in}
}
\subfigure[]{
\label{fig:pred_group_paper}
\epsfig{file=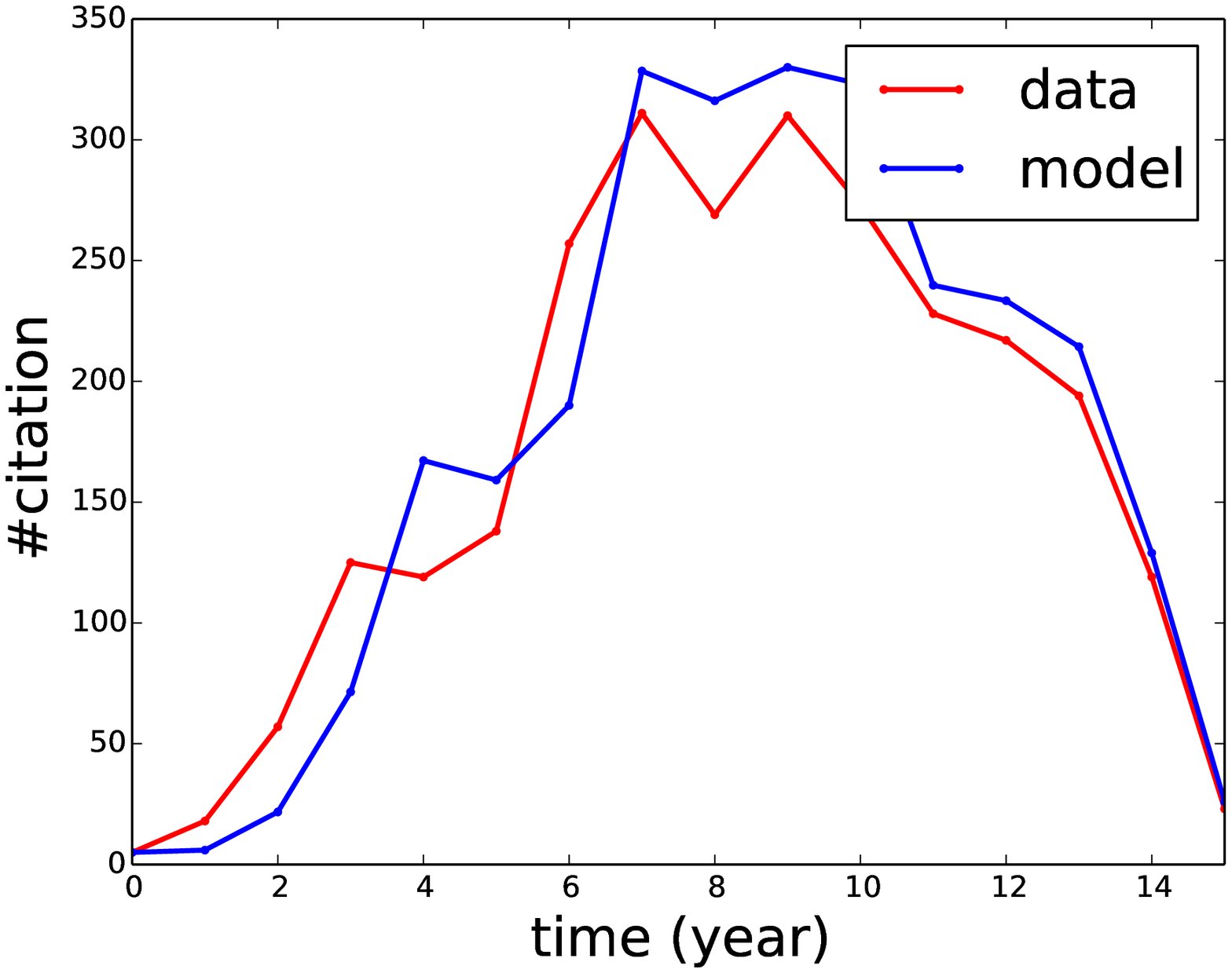, width=2in}
}
\subfigure[]{
\label{fig:pred_group_flickr}
\epsfig{file=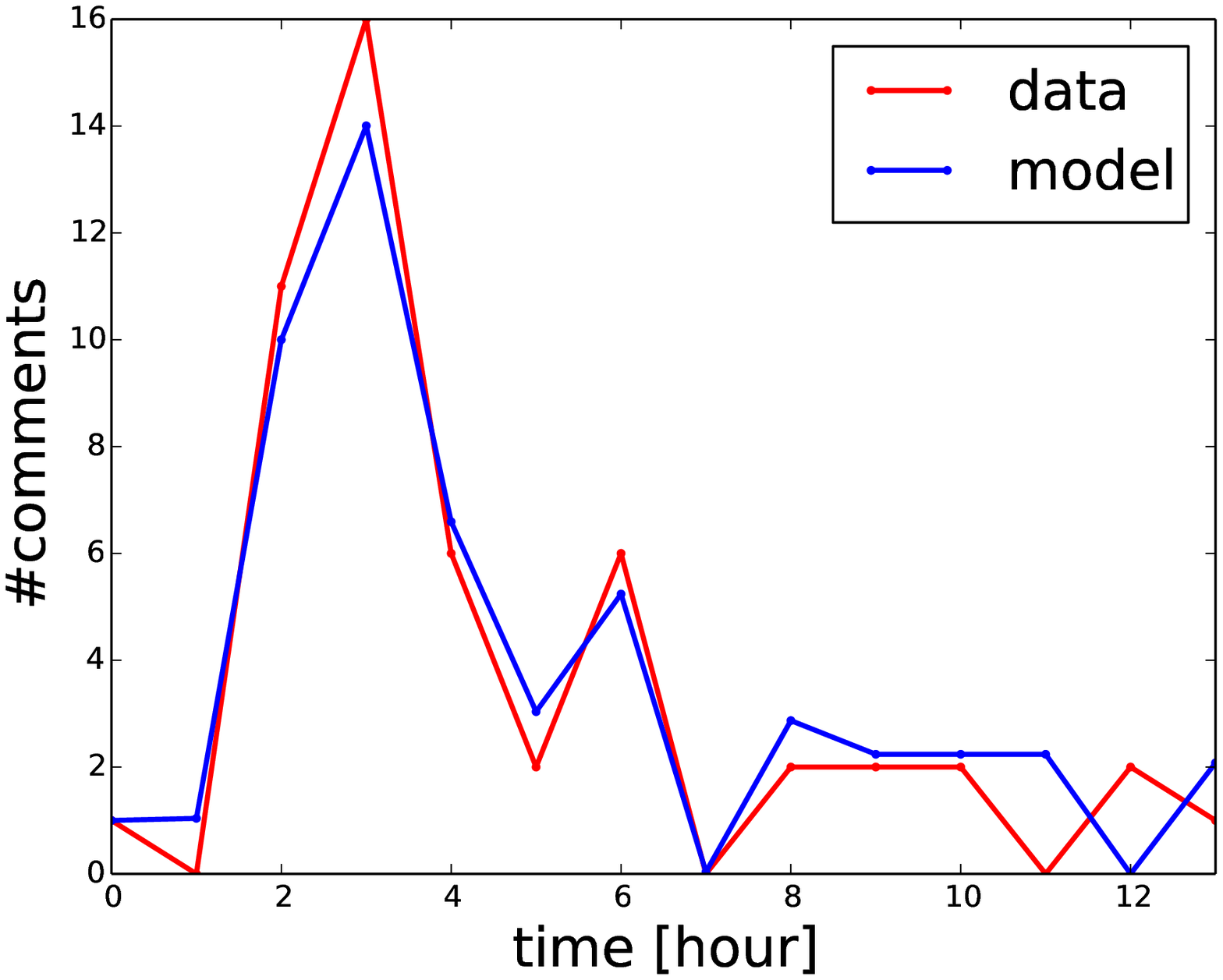, width=2in}
}
}
\caption{Group behavior prediction. In each figure, the x-axis denotes the time that has passed since the source tweet (paper or photo) was published. The y-axis denotes the number of retweets (citations or comments). We demonstrate the observed data (truth) and the prediction results (model). } 
\label{fig:group}
\end{figure*}

\vpara{Results.} Figure~\ref{fig:macro_pred} shows the results on three networks. All figures are plotted on log-log scales. Blue dots represent the distribution of real data. The red solid line is the power-law fitting by observing blue dots through the MLE method~\cite{clauset2009power}. Together with the real data, the fitted result can be used as ground truth. The results of Model I and Model II provided by our framework are represented by the green dashed line and yellow hollow dots, respectively.  


From Figure~\ref{fig:macro_pred}, the power-law distribution suggested by our framework (Model II) is clearly seen to be a good fit to the real data (blue points). 
Please recall that our framework only observes parts of the complete networks ($3.4\%$ in Weibo, $1.3\%$ in Citation, and $4.7\%$ in Flickr).
We also validate the results by quantifying the difference of distributions between Model II and real data via RSS using geometric binning. 
Table~\ref{tb:rss} reports the RSS results of corresponding methods. 
We can see that the performance of our framework can achieve RSS values at the orders of magnitude $10^{-10}$, $10^{-7}$, a nd $10^{-4}$ in the three networks, which indicates excellent numerical performance of the power-law fitting. 
We also notice that the estimated truncation points help our framework (Model II) better fit the heavy-tailed phenomena than Model I. 

We conclude that the emerging process of heavy-tailed phenomena can be modeled and explained well from the 
partially observed individual actions by our framework, on all datasets.

\vpara{Fitting Heavy-Tailed Phenomena using Group Behavior}
Next, we validate whether our proposed \model\ is able to capture the emerging process of heavy-tailed phenomenas from group behavior. Thus we conduct another heavy-tailed phenomena fitting task: Given a set of actions $\{z\}$ and a group behavior $n_z(t)$ for each action $z$ at each timestamp $t$, the goal is to estimate the network distribution $N_z$ that is defined in Definition~\ref{def:network}. 

Following the theoretical results in Section~\ref{sec:framework} and empirical results in Section~\ref{sec:observation}, $n_z(t)$ converges to a lognormal variable when $t$ is sufficiently large. Moreover, \model \ suggests that $N_z$ behaves as power-law when $n_z(t)$ follows lognormal distributions. Thus, our general idea here is to first estimate each timestamp's corresponding lognormal distribution over $n_z(t)$. Then we fit the heavy-tailed phenomena, based on the estimated lognormal. Specifically, for each timestamp $t$, we estimate the corresponding lognormal parameters by following the method introduced in~\cite{zaman2014bayesian}. We then calculate the exponent parameter $\alpha$ according to Theorem~\ref{theorem:powerlaw}. Please notice that we keep $\lambda$, the weighted parameter of the observation window, at the same setting as at the last task when calculating $\alpha$. 

We demonstrate the fitting results in Figure~\ref{fig:macro_pred}, by purple lines. We also present the RSS scores in Table~\ref{tb:rss}. As we can see, comparing with the individual actions (Model II), group behaviors provide more precise information about the lognormal parameters and obtain a better modeling result, which suggests \model \ bridges heavy-tailed phenomena and group behavior precisely.

\subsection{Group Behavior Prediction}
\label{sec:meso}

This task is to demonstrate whether our framework can model the aggregate effect of group behaviors from individual actions in real social networks. 

\vpara{Problem.}
Given an observed subnetwork $G=(V, E)$ of a complete network $H$, a set of individual actions $X_t=\{x_{tvz} | v\in V\}$ at timestamp $t$ in $G$, and $n_z(t)$ in network $H$, the goal is to infer the group behaviors $n_z(t+1)$ at the next timestamp $t+1$. 

\vpara{Setup.}
We separate the observed network $G$ introduced in Section~\ref{sec:exp_macro} into a training set and a test set. 
In Weibo, for each tweet $z$, we regard $\{n_z(t) | t \leq 500 \text{ minutes}\}$ as the training set, and the balance as the test set. 
In Citation, for each paper $z$, we regard $\{n_{z}(t) | t \leq 5 \text{ years} \}$ as the training set, and the rest as the test set. 
In Flickr, for each photo $z$, we use $\{n_{z}(t) | t \leq 2 \text{hours}\}$ as the training set. 
We then employ the training set to estimate the ``upward factor'' $U$ and ``downward factor'' $D$ of the group model in our proposed framework (see details in Eq.~\ref{eq:meso}). 
We finally calculate $n(t+1)$ by using $n_z(t)$, $U$, and $D$ according to Eq.~\ref{eq:meso}. 



There are several methods for estimating factors $U$ and $D$. Here, we assume that the factors are constant over both users and time. Under this assumption, given any two timestamps $t_1$ and $t_2$, we are able to estimate the factors according to equations below:

\beqn{
\label{eq:estimate_factor}
\tiny
\besp{
\ln \widehat{U}_{t_1t_2} &=& \frac{(\ln n(t_1+1) - \ln n(t_1))y_{t_2}^-+(\ln n(t_2) - \ln n(t_2+1))y_{t_1}^-}{y_{t_1}^+y_{t_2}^--y_{t_2}^+y_{t_1}^-}\\
\ln \widehat{D}_{t_1t_2} &=&\frac{(\ln n(t_1+1)-\ln n(t_1))y_{t_2}^+-(\ln n(t_2)-\ln n(t_2+1))y_{t_1}^+}{y_{t_1}^-y_{t_2}^+-y_{t_2}^-y_{t_1}^+}
}
\normalsize
} 

Please note that for different pairs of timestamps, the estimated results might be different. Hence, we use the average value of all possible configurations as the final reported results. Formally, let $\widehat{U}_{t_1t_2}$ and $\widehat{D}_{t_1t_2}$ be the factors estimated according to time $t_1$ and $t_2$ by Eq.~\ref{eq:estimate_factor}, we define $\widehat{U}=\frac{\sum_{t_1,t_2} \widehat{U}_{t_1t_2}}{C_T^2}$ and $\widehat{D}=\frac{\sum_{t_1,t_2} \widehat{D}_{t_1t_2}}{C_T^2}$, 

\hide{
\beqn{
\besp{
\widehat{U}=\frac{\sum_{t_1,t_2} \widehat{U}_{t_1t_2}}{C_T^2}  \\
\widehat{D}=\frac{\sum_{t_1,t_2} \widehat{D}_{t_1t_2}}{C_T^2}
}
}
}
\noindent where $T$ is the last timestamp in the training set.   

In practice, an alternative method is to assign each user different configurations of these two factors. However, our main goal in this paper is to provide the underlying mechanisms of the emergence of heavy-tailed phenomena. Thus we keep the assumption that the two factors are independent of both time and users, to simplify and generalize our proposed framework. We leave the user- or time-dependent factor definition and estimation for our future work.


\vpara{Results.} 
We report the experimental results of group behavior prediction in Figure~\ref{fig:group} in the three networks. 
In each figure, the $x$-axis denotes the duration since the source tweet (paper or photo) is posted by setting ten minutes (one year or one day) as the interval in Weibo (Citation or Flickr). 
The $y$-axis denotes the number of retweets (citations or comments) the source tweet (paper or photo) receives at each time interval. 
We plot the truth in real networks using red lines and the predicting results of \model \ with blue lines. 

Figure~\ref{fig:pred_group_post} presents the results for tweets with different popularity (indicated by the number of retweets) in Weibo. 
Clearly, we can see that the modeling results (blue lines) are well coupled with the real data (red lines) in different cases. 
Our method successfully captures the upward and downward tendencies of each tweet's retweeting dynamics over time precisely. 


Figure~\ref{fig:pred_group_paper} shows the results of two papers with different levels of citation counts (2262 and 122). 
Figure~\ref{fig:pred_group_flickr} shows the results of two photos with different numbers of comments (50 and 133). 
As in Weibo, group behaviors can be successfully inferred from individual citing and commenting actions in Citation and Flickr. 

We conclude that our framework can capture the aggregate effect of group behaviors from individual actions, and predict the trends of dynamic popularities in real social networks.



\subsection{Information Burst Prediction}
\label{sec:burst}

We now describe ways in which to apply our framework to social applications. 
In this work, we focus on information burst prediction~\cite{Kleinberg:KDD2002Bursty,Myers:WWW2014}. 
Please note that the focus of the study is to demonstrate how our framework can help social applications.


\vpara{Problem.}
Given a tweet $z$ at timestamp $t$ and the number of users who retweet $z$ within the time window $[t-k, t]$ ($\{n_z(t') | t-k\le t' \leq t\}$), the goal is to predict whether there will be an information burst at timestamp $t+1$. Formally, we say a burst happens at time $t_1$ if $n_z(t_1)$ is the largest in the period ranging from one hour before to one hour after $t_1$, i.e., $\forall t_2 \in [t_1-1\textit{hr}, t_1+1\text{hr}]$, we have $n_z(t_1) > n_z(t_2)$. 
\hide{ 
\reminder{ what is $\Delta$t?
We say a burst happens at time $t_1$ if $n_z(t_1)$ is the largest one within time window $[t_1-\Delta t, t_1+\Delta t]$, i.e., $\forall t_2 \in [t-\Delta t, t+\Delta t]$, we have $n_z(t_1) \ge n_z(t_2)$. In our experiment, we set $\Delta t$ as 60 minutes. 
}}

\begin{figure}[t]
\centering
\hspace{-0.3in}
\mbox{
\subfigure[]{
\label{fig:factor_analysis}
\epsfig{file=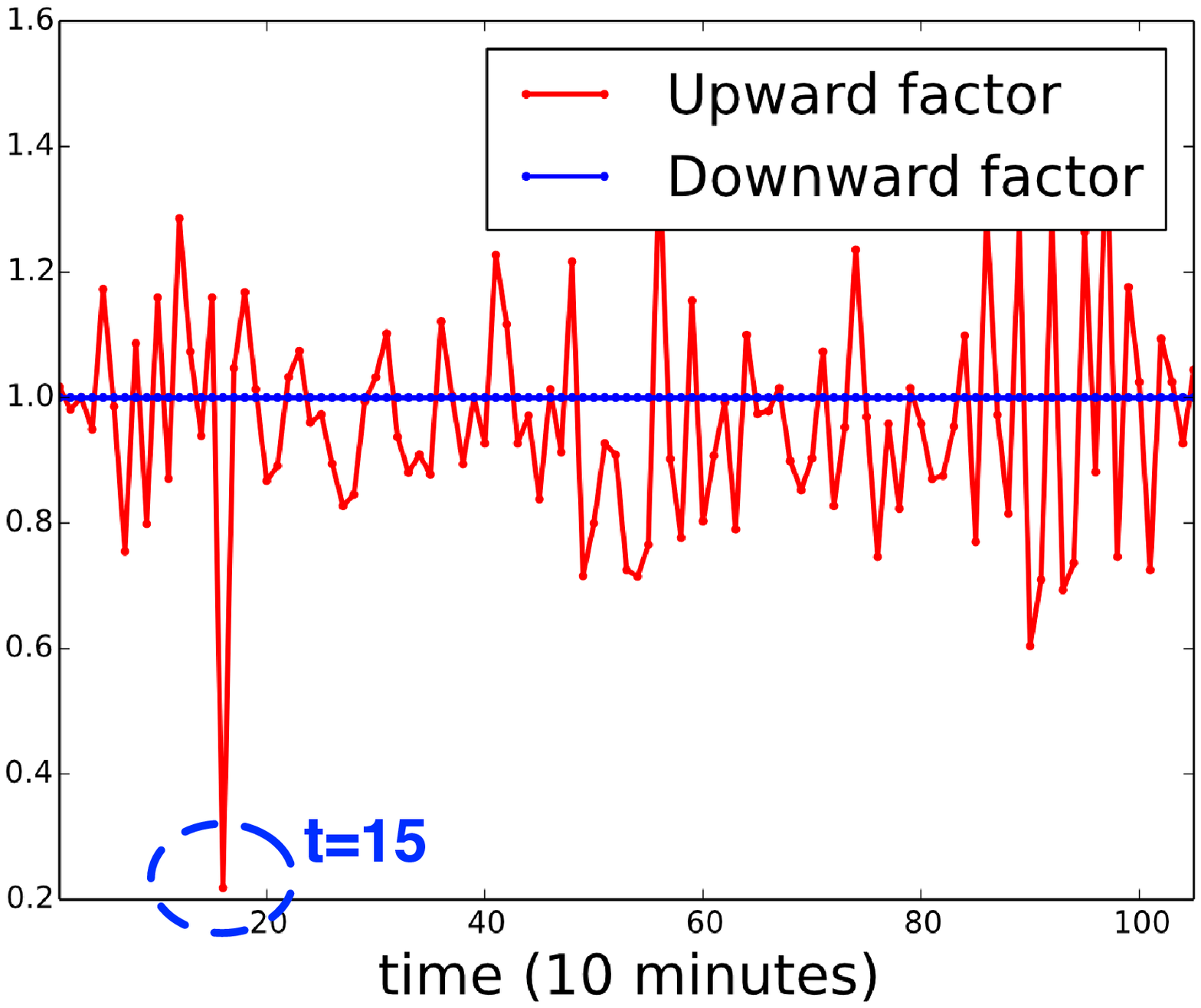, width=1.7in}
}
}
\hspace{-0.12in}
\mbox{
\subfigure[]{
\label{fig:burst_example}
\epsfig{file=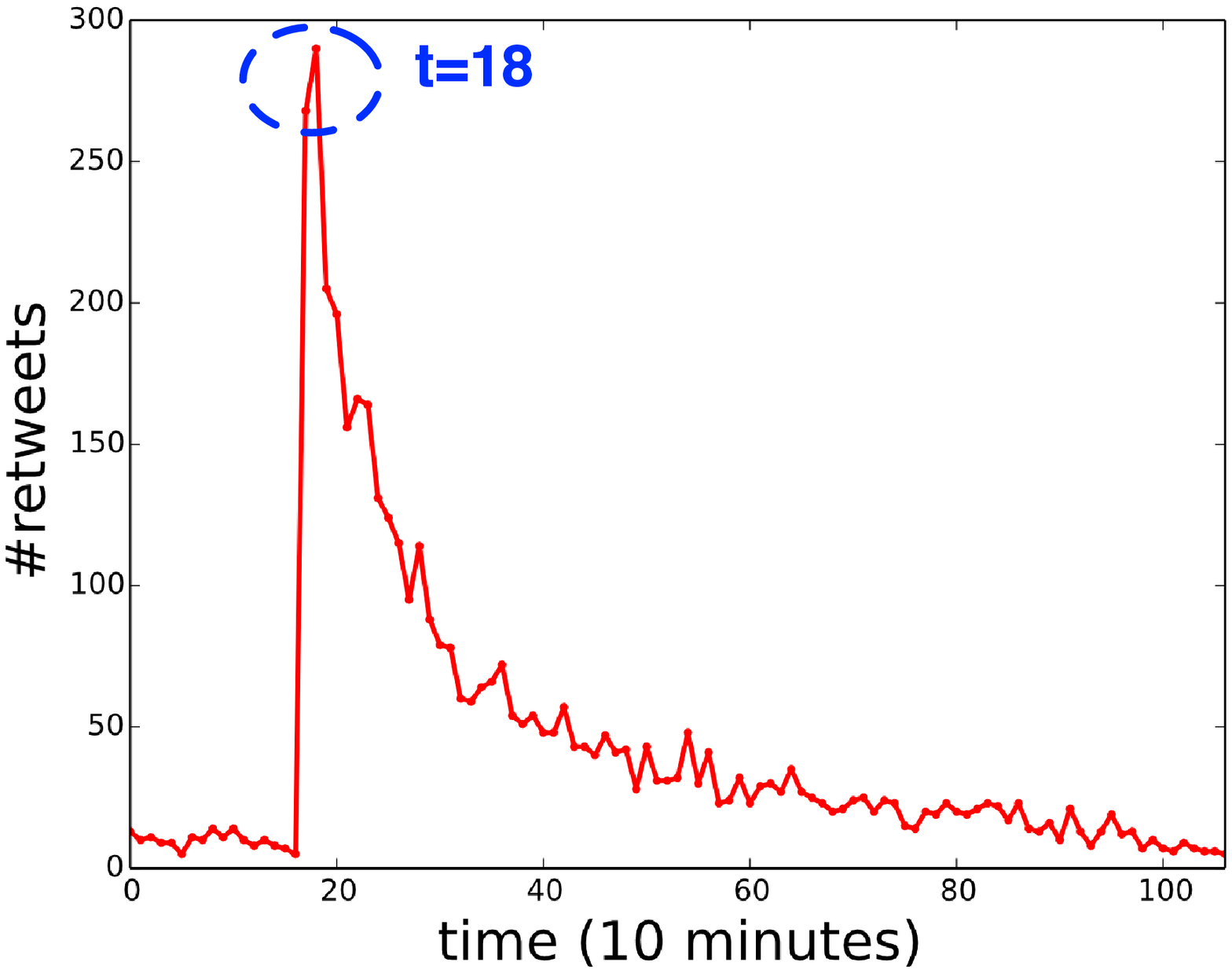, width=1.7in}
}
}
\hspace{-0.3in}
\caption{Factor analysis. (a) The trend of the upward factor and the downward factor, corresponding to the tweet in (b), over time. (b) An example of information burst.}
\label{fig:factor} 
\end{figure}

\vpara{Observation.}
In practice, the upward factor $U$ and downward factor $D$, instead of being constant, may change over time. We study the trends of these factors and present the results in Figure~\ref{fig:factor}. 
Due to space limitations, we use the factors corresponding to one tweet as an example. We observe similar results on other tweets. 

In Figure~\ref{fig:factor_analysis}, we can see that the downward factors are relatively stable around $1.0$, while the upward factors change frequently and sharply. 
A potential explanation is that a user's retweeting actions can influence others to retweet, while the decision not to retweet a tweet has limited effect on others' retweeting decisions. 
We further observe that when observing a peak or a valley in Figure~\ref{fig:factor_analysis}, the upward factor will achieve a peak in next few timestamps in Figure~\ref{fig:burst_example}. 
We conjecture that the burst of retweets is correlated with the variation of the corresponding upward factors. 


\vpara{Setup.}
We apply the upward factors in our framework into the information burst prediction task. 
We first estimate $\widehat{U}_{t'}$ at each timestamp $t'$ in the observation time window $[t-k+1, t]$. 
We then use the estimated upward factors as features of classification models. 
We also consider a baseline method in which the number of retweets in previous steps $\{n_z(t')|t-k \le t' \leq t\}$ is used as the features for this task. 
As our goal is to provide evidence of the correspondence between upward factors and information burst, we simply use classic classifiers, logistic regression (LR) and support vector machine (SVM), to report the predictability. 


\vpara{Results.}
Figure~\ref{fig:peak_pred} shows the predictive performance. 
We can clearly see that by our methodology, both simple models significantly outperforms the comparison method (using previous retweets information---$n_z(\cdot)$ as features). 
Moreover, the methods that involve using upward factors derived from our framework can achieve an F1 score of 0.8, which demonstrates the predictability of bursty phenomena in social media.  
In terms of precision and recall, the performance is still promising.

We further examine how the length of the observation window influences the prediction performance in Figure~\ref{fig:peak_pred}. 
We observe that both LR and SVM achieve the best and stable performance when the observation window reaches 200 minutes (around three hours). 
We conclude that the predictability of bursty phenomena in information diffusion is highly correlated to the observation window, and information burst is more predictable when conducted over a sufficient timeframe---only three hours in Weibo. 

Further study of other applications, such as cascade prediction, scientific impact modeling, and popularity forecast, is an area of work for the near future.

\begin{figure}[t]
\centering
\hspace{-0.3in}
\mbox{
\subfigure[]{
\label{fig:pred_feature}
\epsfig{file=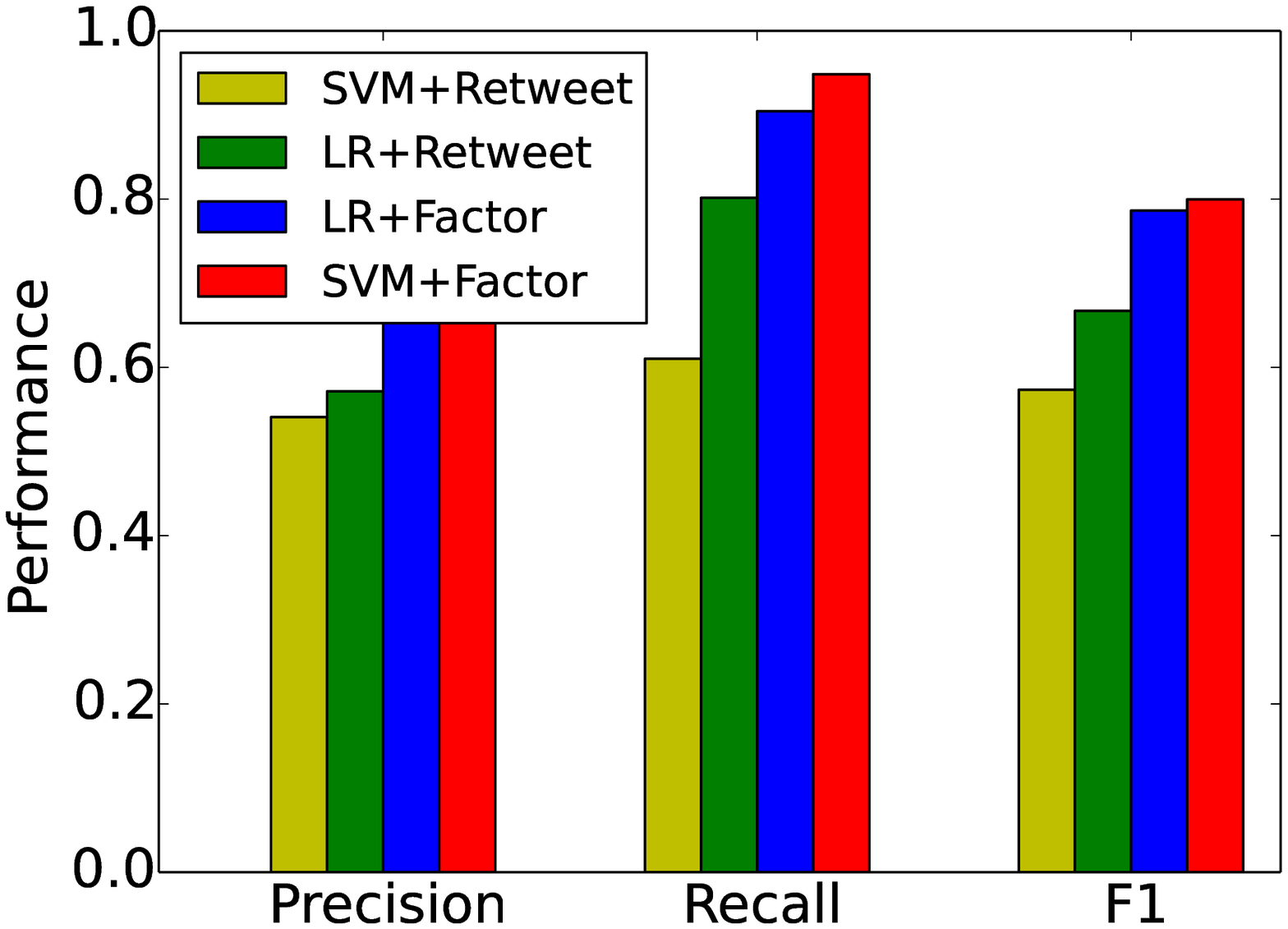, width=1.7in}
}
}
\hspace{-0.12in}
\mbox{
\subfigure[]{
\label{fig:peak_pred}
\epsfig{file=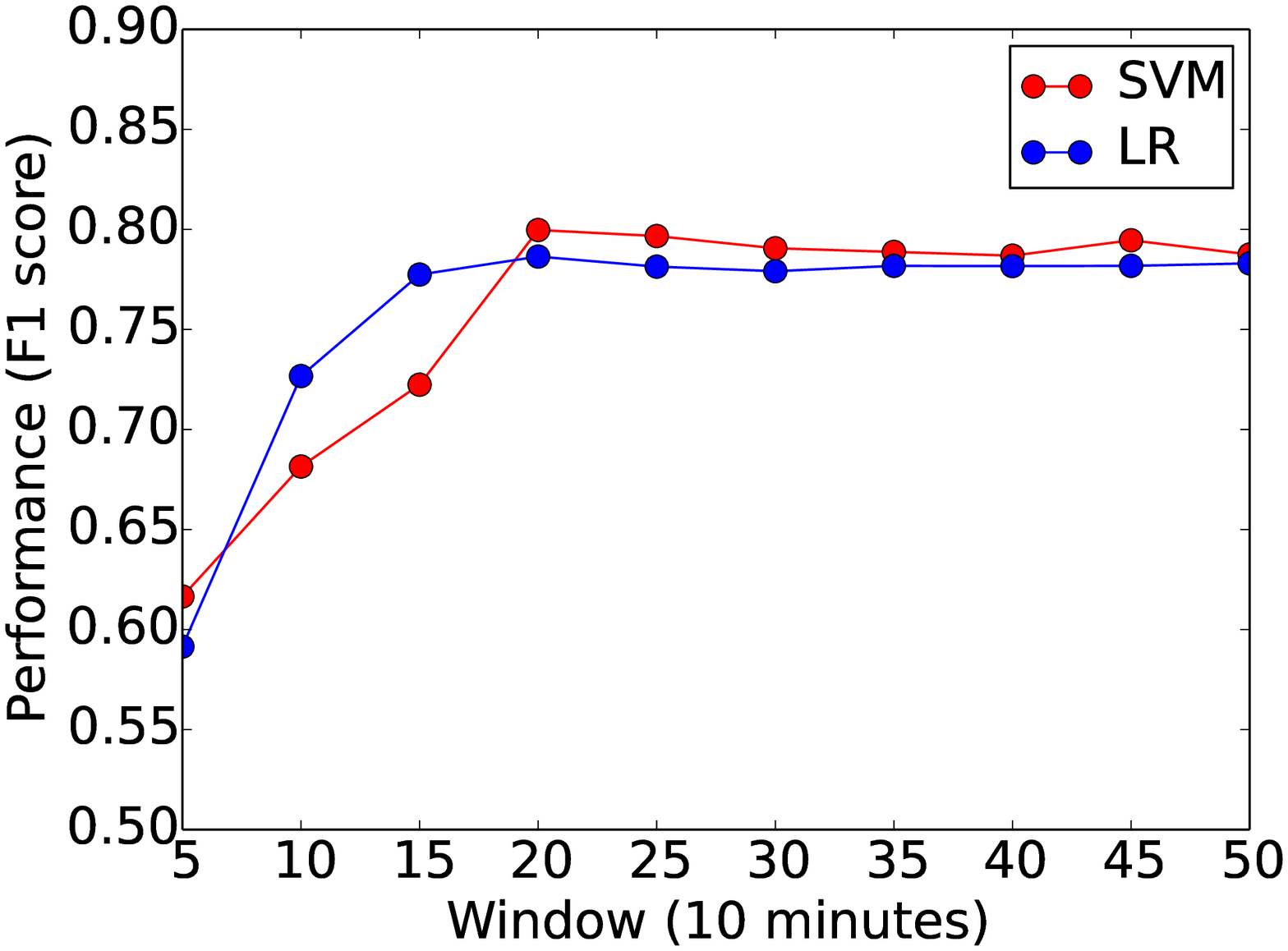, width=1.7in}
}
}
\hspace{-0.3in}
\caption{Apply \model \ to information burst prediction. (a) Prediction performance by two classification models considering historic retweet times (SVM+Retweet and LR+Retweet) and models utilizing the upward factors in \model \ (SVM+Factor and LR+Factor). We set the observation window $k$ as 200 minutes. (b) Prediction performance by varying the observation windows from 50 minutes to 500 minutes. \normalsize} 
\label{fig:peak}
\end{figure}


\section{Related Work}
\label{sec:related}

%
The heavy-tailed phenomena---such as power-law and lognormal distributions---have been discovered to be ubiquitous in a variety of network systems~\cite{mood1950introduction,newman2005power,clauset2009power}. 

\textbf{Power-laws} have been widely observed in both nature and human society through extensive studies. 
Power-laws are characterized by the following probability distribution: 
\beq{
\label{eq:power_law}
f(x) =C x^{-\alpha} \nonumber
}
\noindent where $\alpha$ is the exponent parameter and $C$ is the normalization term. 
Essentially, power-laws model the functional relationships between two quantities, where one quantity varies as a power of another~\cite{newman2005power}. 
This statistical law was first revealed by the degree distributions of Internet graphs and the World Wide Web in 1999~\cite{faloutsos1999power,barabasi1999emergence}. 
Besides having been found in Internet and WWW networks, power-law distributions have been discovered in publication citations~\cite{Radicchi:PNAS2008}, phone calls~\cite{Christos:KDD08}, tie strengths~\cite{Onnela:PNAS2007}, and so on. 

Along with power-laws, extensive studies have discovered that \textbf{lognormal} distributions are satisfied in dynamic networks. 
Conceptually, a lognormal distribution is defined as ``a continuous probability distribution of a random variable whose logarithm is normally distributed.''~\cite{mood1950introduction}. 
Formally, its density function can be expressed by the following formula: 
\beq{
\label{eq:lognormal}
f(x)=\frac{1}{x\delta\sqrt{2\pi}}\exp[-\frac{(\ln x - \mu)^2}{2\delta^2}] \nonumber
}
\noindent where $\mu$ and $\delta$ are, respectively, the mean and standard deviation of the variable's natural logarithm. 
Huberman et al.~\cite{huberman1999internet} presented that the number of pages at a given site is lognormally distributed for every timestmap in the environment of WWW. 
Following this, similar studies also unveiled that part of the WWW pages demonstrate a truncated lognormal distribution~\cite{Bi:KDD2001,Pennock:PNAS02}. 
Stouffer et al.~\cite{stouffer2006log} discovered that the time for people to respond to emails also follows a lognormal distribution. 
Although a tremendous amount of exploration into the modeling of heavy-tailed phenomena in networks has been done, the mechanism of how these phenomena emerge from individual actions has  received little attention. 

An individual behavior, such as an action adoption, is an integral element of the heavy-tailed phenomena in social networks. 
A large body of studies has been focused on the modeling and predicting of individual behaviors. 
Xu, and Hong et al.~\cite{Xu:SIGIR2012,Hong:WSDM2013} modeled user online preferences and predicted individual adoption decisions in Twitter.  
However, 
the connections between individual behavior and collective dynamics are still not well studied. 
Recently, Rybski et al.~\cite{Rybski:SR12} studied individual behaviors and further unveiled the origin of collective behaviors of the social community with both clustering and long-term persistence. 
Muchnik et al.~\cite{Muchnik:SR14} demonstrated that heavy-tailed degree distributions in networks are causally determined by similarly skewed distributions of human activity. 
Ghosh et al.~\cite{ghosh2014interplay} studied the interplay between a dynamic process and the structure of the network on which it is defined. 
The major difference between our work and previous work lies in that
we theoretically and empirically demonstrate how the integration of individual actions in social networks activates the emergence of group behavior, from which network distributions arise as a whole.


\section{Conclusion}
\label{sec:conclusion}

In this paper, we study a novel problem of modeling the interplay between individual behavior and network distributions. We propose a unified framework \model\ to model individual behavior and network distributions together. The framework offers a way to explain how group behavior has  developed and evolves over time based on individual  actions, and to understand how network distributions such as power-law (or heavy-tailed phenomena) can be explained by group behavioral patterns.
The framework is flexible and can benefit many applications. We apply \model\ to three different networks: Tencent Weibo, Citation, and Flickr.
Our experimental results show that \model \ is able to model emerging network distributions in social networks from individual actions.
Moreover, we use information burst prediction as an application to quantitatively evaluate the predictive power of \model.

\balance
\bibliographystyle{abbrv}
\bibliography{reference}  

\begin{thebibliography}{10}

\bibitem{adamic1999nature}
L.~Adamic and B.~A. Huberman.
\newblock The nature of markets in the world wide web.
\newblock {\em Q. J. Econ.}, 1999.

\bibitem{akaike1998information}
H.~Akaike.
\newblock Information theory and an extension of the maximum likelihood
  principle.
\newblock In {\em Selected Papers of Hirotugu Akaike}, pages 199--213. 1998.

\bibitem{barabasi1999emergence}
A.-L. Barab{\'a}si and R.~Albert.
\newblock Emergence of scaling in random networks.
\newblock {\em science}, 286(5439):509--512, 1999.

\bibitem{Bernheim:94}
B.~D. Bernheim.
\newblock A theory of conformity.
\newblock {\em J. Polit. Econ.}, 1027(5):841--877, 1994.

\bibitem{Bi:KDD2001}
Z.~Bi, C.~Faloutsos, and F.~Korn.
\newblock The "{DGX}" distribution for mining massive, skewed data.
\newblock In {\em KDD '01}, pages 17--26, 2001.

\bibitem{black1973pricing}
F.~Black and M.~Scholes.
\newblock The pricing of options and corporate liabilities.
\newblock {\em J. Polit. Econ.}, pages 637--654, 1973.

\bibitem{clauset2009power}
A.~Clauset, C.~R. Shalizi, and M.~E. Newman.
\newblock Power-law distributions in empirical data.
\newblock {\em SIAM review}, 51(4):661--703, 2009.

\bibitem{draper2014applied}
N.~R. Draper and H.~Smith.
\newblock {\em Applied regression analysis}.
\newblock 2014.

\bibitem{faloutsos1999power}
M.~Faloutsos, P.~Faloutsos, and C.~Faloutsos.
\newblock On power-law relationships of the internet topology.
\newblock In {\em COMPUT COMMUN REV}, volume~29, pages 251--262, 1999.

\bibitem{ghosh2014interplay}
R.~Ghosh, S.-H. Teng, K.~Lerman, and X.~Yan.
\newblock The interplay between dynamics and networks: centrality, communities,
  and cheeger inequality.
\newblock In {\em KDD'14}, pages 1406--1415, 2014.

\bibitem{Hong:WSDM2013}
L.~Hong, A.~S. Doumith, and B.~D. Davison.
\newblock Co-factorization machines: Modeling user interests and predicting
  individual decisions in twitter.
\newblock In {\em WSDM '13}, pages 557--566, 2013.

\bibitem{huberman1999internet}
B.~A. Huberman and L.~A. Adamic.
\newblock Internet: growth dynamics of the world-wide web.
\newblock {\em Nature}, 401(6749):131--131, 1999.

\bibitem{Kelman:58}
H.~C. Kelman.
\newblock Compliance, identification, and internalization: Three processes of
  attitude change.
\newblock {\em J. Confl. Resolut.}, 2(1):51--60, 1958.

\bibitem{Kleinberg:KDD2002Bursty}
J.~Kleinberg.
\newblock Bursty and hierarchical structure in streams.
\newblock In {\em KDD '02}, pages 91--101, 2002.

\bibitem{lou2013mining}
T.~Lou and J.~Tang.
\newblock Mining structural hole spanners through information diffusion in
  social networks.
\newblock In {\em WWW'13}, pages 825--836, 2013.

\bibitem{mitzenmacher2004brief}
M.~Mitzenmacher.
\newblock A brief history of generative models for power law and lognormal
  distributions.
\newblock {\em Internet mathematics}, 1(2):226--251, 2004.

\bibitem{mood1950introduction}
A.~M. Mood.
\newblock Introduction to the theory of statistics.
\newblock 1950.

\bibitem{Muchnik:SR14}
L.~Muchnik, S.~Pei, L.~C. Parra, S.~D. Reis, J.~S. Andrade~Jr, S.~Havlin, and
  H.~A. Makse.
\newblock Origins of power-law degree distribution in the heterogeneity of
  human activity in social networks.
\newblock {\em Scientific reports}, 3, 2013.

\bibitem{Myers:WWW2014}
S.~A. Myers and J.~Leskovec.
\newblock The bursty dynamics of the twitter information network.
\newblock In {\em WWW '14}, pages 913--924, 2014.

\bibitem{newman2005power}
M.~E. Newman.
\newblock Power laws, pareto distributions and zipf's law.
\newblock {\em Contemporary physics}, 46(5):323--351, 2005.

\bibitem{Onnela:PNAS2007}
J.-P. Onnela, J.~Saram\"aki, J.~Hyv\"onen, G.~Szab\'{o}, D.~Lazer, K.~Kaski,
  J.~Kert\'{e}sz, and A.-L. Barab\'{a}si.
\newblock Structure and tie strengths in mobile communication networks.
\newblock {\em PNAS}, 104(18):7332--7336, 2007.

\bibitem{Pennock:PNAS02}
D.~M. Pennock, G.~W. Flake, S.~Lawrence, E.~J. Glover, and C.~L. Giles.
\newblock Winners don't take all: Characterizing the competition for links on
  the web.
\newblock {\em PNAS}, 99(8):5207--5211, 2002.

\bibitem{Radicchi:PNAS2008}
F.~Radicchi, S.~Fortunato, and C.~Castellano.
\newblock Universality of citation distributions: Toward an objective measure
  of scientific impact.
\newblock {\em PNAS}, 2008.

\bibitem{Rybski:SR12}
D.~Rybski, S.~V. Buldyrev, S.~Havlin, F.~Liljeros, and H.~A. Makse.
\newblock Communication activity in a social network: relation between
  long-term correlations and inter-event clustering.
\newblock {\em Scientific reports}, 2, 2012.

\bibitem{Christos:KDD08}
M.~Seshadri, S.~Machiraju, A.~Sridharan, J.~Bolot, C.~Faloutsos, and
  J.~Leskove.
\newblock Mobile call graphs: Beyond power-law and lognormal distributions.
\newblock In {\em KDD '08}, pages 596--604, 2008.

\bibitem{Song:Connect2012}
C.~Song, D.~Wang, and A.-L. Barabasi.
\newblock Connections between human dynamics and network science.
\newblock {\em arXiv preprint arXiv:1209.1411}, 2012.

\bibitem{stouffer2006log}
D.~B. Stouffer, R.~D. Malmgren, and L.~A. Amaral.
\newblock Log-normal statistics in e-mail communication patterns.
\newblock {\em arXiv preprint physics/0605027}, 2006.

\bibitem{Tang:08KDD}
J.~Tang, J.~Zhang, L.~Yao, J.~Li, L.~Zhang, and Z.~Su.
\newblock Arnetminer: extraction and mining of academic social networks.
\newblock In {\em KDD'08}, pages 990--998, 2008.

\bibitem{vazquez2006modeling}
A.~V{\'a}zquez, J.~G. Oliveira, Z.~Dezs{\"o}, K.-I. Goh, I.~Kondor, and A.-L.
  Barab{\'a}si.
\newblock Modeling bursts and heavy tails in human dynamics.
\newblock {\em Physical Review E}, 73(3):036127, 2006.

\bibitem{Wilk:QQ1968}
M.~B. Wilk and R.~Gnanadesikan.
\newblock {Probability plotting methods for the analysis of data.}
\newblock {\em Biometrika}, 55(1):1--17, Mar. 1968.

\bibitem{Wu:WWW2011}
S.~Wu, J.~M. Hofman, W.~A. Mason, and D.~J. Watts.
\newblock Who says what to whom on twitter.
\newblock In {\em WWW '11}, pages 705--714, 2011.

\bibitem{Xu:SIGIR2012}
Z.~Xu, Y.~Zhang, Y.~Wu, and Q.~Yang.
\newblock Modeling user posting behavior on social media.
\newblock In {\em SIGIR '12}, pages 545--554, 2012.

\bibitem{yang2014your}
Y.~Yang, J.~Jia, S.~Zhang, B.~Wu, Q.~Chen, J.~Li, C.~Xing, and J.~Tang.
\newblock How do your friends on social media disclose your emotions?
\newblock In {\em AAAI'14}, 2014.

\bibitem{yang2014rain}
Y.~Yang, J.~Tang, C.~W.-k. Leung, Y.~Sun, Q.~Chen, J.~Li, and Q.~Yang.
\newblock Rain: Social role-aware information diffusion.
\newblock In {\em AAAI'15}, 2014.

\bibitem{zaman2014bayesian}
T.~Zaman, E.~B. Fox, E.~T. Bradlow, et~al.
\newblock A bayesian approach for predicting the popularity of tweets.
\newblock {\em AOAS}, 8(3):1583--1611, 2014.

\end{thebibliography}


\normalsize

\end{document}